\documentclass{SLC}
\articlenumber{12} 
\usepackage{bbm} %for mathbbm 1
\usepackage{enumitem}
\usepackage{array,tabularx}
\graphicspath{{pics/}}
\usepackage{etoolbox}\AtBeginEnvironment{pmatrix}{\setlength{\arraycolsep}{1pt}}

\usepackage[american]{babel} 
\usepackage[utf8]{inputenc}
\usepackage{cite} %ordering the cites
\usepackage[all]{hypcap}
\usepackage{caption}

\newcommand{\law}{\ensuremath{\stackrel{d}=}}
\def\sm{\scalebox{0.85}[1.11]{$\,-\,$}}
\def\sp{\scalebox{0.95}{$\,+\,$}}
\renewcommand{\Pr}{\mathbb{P}} 

\newcommand\indicator{{\mathbbm 1}}
\newcommand\E{\mathbb{E}}
\newcommand\Var{\mathbb{V}{\rm ar}}
\newcommand\phim{\phi}
\newcommand\phiv{\phi_{\rm var}}
\newcommand{\indi}{{\rm 1\!I}}
\newcommand\N{\mathbb{N}}
\newcommand\R{\mathbb{R}}
\newcommand\Z{\mathbb{Z}}
\renewcommand\S{\mathcal{S}}
\def\X{\mathbf X}
\def\Y{\mathbf Y}
\def\Fh{F^{\leq h}}
\def\Fh{F^{\leq h}}
\def\H{{\widetilde H}}
\def\errortermLNthree{O\left(\frac{(\ln n)^3}{n}\right)}
\def\errortermLNfour{O\left(\frac{(\ln n)^4}{n}\right)}
\def\errorterm{\errortermLNthree}
\def\Res{\operatorname{Res}}

\def\C{Z_1} %Principal contribution to the cumulative density function of \Fh(z,1)
\title[Height of walks with resets, the Moran model, and the discrete Gumbel distribution]
{Height of walks with resets, the Moran model, and the discrete Gumbel distribution}

\author[R.~Aguech, A.~Althagafi, and C.~Banderier]{Rafik Aguech\textsuperscript{1}\protect\orcid{0000-0002-4483-9356}}
\author[]{Asma Althagafi\textsuperscript{2}\protect\orcid{0000-0001-5499-0810}}
\author[]{Cyril Banderier\textsuperscript{3}\protect\orcid{0000-0003-0755-3022}}
\address{\textsuperscript{1} {Department of Statistics and Operations Research, King Saud Univ., Saudi Arabia,
and Department of Mathematics, University of Monastir, Tunisia;\protect\newline\website{https://faculty.ksu.edu.sa/en/raguech}}}
\address{\textsuperscript{2} {Department of Statistics and Operations Research, King Saud Univ., Saudi Arabia;
\website{https://www.researchgate.net/profile/Asma-Althagafi}}}
\address{\textsuperscript{3} Laboratoire d'Informatique de Paris Nord, Univ.\ Sorbonne Paris Nord, France; \website{http://lipn.fr/~banderier}\bigskip\bigskip} 

\abstract{In this article, we consider several models of random walks in one or several dimensions,
additionally allowing, at any unit of time, a reset (or ``catastrophe'') of the walk with probability $q$.
We establish the distribution of the final altitude.
We prove algebraicity of the generating functions of walks of bounded height~$h$
(showing in passing the equivalence between Lagrange interpolation and the kernel method).
To get these generating functions, our approach offers an algorithm of cost $O(1)$, instead of cost $O(h^3)$ if a Markov chain approach would be used.
The simplest nontrivial model corresponds to famous dynamics in population genetics: the Moran model.

We prove that the height of these Moran walks asymptotically follows a discrete Gumbel distribution.
For $q=1/2$, this generalizes a model of carry propagation over binary numbers considered e.g.~by von Neumann and Knuth.
For generic $q$, using a Mellin transform approach, 
we show that the asymptotic height exhibits fluctuations for which we get an explicit description (and, in passing, new bounds for the digamma function).
We end by showing how to solve multidimensional generalizations of these walks (where any subset of particles is attributed a different probability of dying)
and we give an application to the soliton wave model.\bigskip}


\begin{document}
\maketitle
\pagebreak
\tableofcontents

\vfill
\section{Introduction}
The height of random walks is a fundamental parameter which occurs in many domains: 
in computer science (evolution of a stack, tree traversals, or cache algorithms~\cite{Knuth1997}),
in reliability or failure theory (maximal age of a component and inference statistics on the longevity before replacement~\cite{Gertsbakh1989}), 
in queueing theory (maximal length of the queue, with e.g.~applications to traffic jam analysis~\cite{Kerner2009}), 
in mathematical finance (e.g.~in risk theory~\cite{Grandell1991}),
in bioinformatics (pattern matching and sequence alignment~\cite{AltschulKarlin1990}), etc.

In combinatorics, random walks are studied via the corresponding notion of \textit{lattice paths},
which play a central role, not only for intrinsic properties of such paths, 
but also as they are in bijection with many fundamental structures
(trees, words, maps, \dots). We refer to the nice \textit{magnum opus} of Flajolet and Sedgewick 
on analytic combinatorics~\cite{FlajoletSedgewick2009} for many enumerative and asymptotic 
examples.

While the behavior of an extremal parameter such as the height is well understood for walks corresponding to Brownian motion theory,
it becomes more subtle when a notion of reset/renewal/resetting/catastrophe~\cite{BanderierWallner2017b,Krinik07,Krinik08,BenAriRoitershteinSchinazi17,HarrisTouchette17,Schehr15,Majumdar14}
is introduced in the model: indeed, typical behaviors in this model
are often established by conditioning on events of probability zero in the model without reset, leading to possibly counterintuitive results.

\pagebreak 

In this article, we give several enumerative and asymptotic results on different statistics (final altitude, waiting time, height) of walks with resets,
focusing on the so-called Moran walks (walks related to biological/population models considered by Moran in 1958; see Section~\ref{Sec5} for more on this).

%Such resets occur naturally in the applications listed above. For example, in molecular biology, 
%random walks are natural models for the evolution of a pattern matching algorithms on DNA sequences,
% and a a reset would model the start/stop codons, or a discrepancy threshold.
%Algorithms using the principle of some BLAST heuristics (basic local alignment search tool) assign a score to some sequences of codons:
% this can then be seen as the height of a lattice paths; see e.g.~the work of Altschul and Karlin~\cite{AltschulKarlin1993}.

\bigskip

\textbf{Plan of the article.}

In Section~\ref{Sec2}, we consider a generic model of walks with resets (allowing any finite set of steps and a reset step).
We describe the behavior of their final altitude (at finite time, and asymptotically).
We obtain an algebraic closed form for the bivariate generating function (length/final altitude) for walks of bounded height~$h$.
Our approach uses a variant of the so-called kernel method, 
which has the advantage to avoid any case-by-case computation based on Markov chains/transfer matrices of size $h \times h$.
In passing, we show the intimate link between Lagrange interpolation and the kernel method.

In Section~\ref{Sec3}, we consider Moran walks,
a model described in~Figure~\ref{Fig:Moran}, for which we generalize an enumerative formula due to Pippenger~\cite{Pippenger2002}. 
We show that their height asymptotically follows a distribution which involves non-trivial fluctuations.
We prove that this distribution is a \textit{discrete} Gumbel distribution, and we clarify its links with the \textit{continuous} Gumbel distribution.
We give an application to the waiting time for reaching any given altitude.

In Section~\ref{Sec4}, we begin with a brief presentation of the Mellin transform method,
and then use it to derive a precise analysis of the asymptotic average and variance of the height.
The second asymptotic term involves some $O(1)$ fluctuations given by a Fourier series
(which we prove to be infinitely differentiable, and for which we also derive generic bounds of independent interest).
This extends (and fixes some error terms) in earlier analyses by von Neumann, Knuth, Flajolet and Sedgewick~\cite{BurksGoldstinevonNeumann1946,Knuth1978,FlajoletSedgewick2009}.

In Section~\ref{Sec5}, we tackle some multidimensional generalizations of Moran walks, with applications to a model in population genetics 
and to a wave propagation model (a soliton model), as considered by Itoh, Mahmoud, and Takahashi~in~\cite{ItohMahmoudTakahashi2004,ItohMahmoud2005}. 

In Section~\ref{Sec6}, we conclude with a few possible extensions for future work.

\begin{figure}[h]
\includegraphics[width=.9\textwidth]{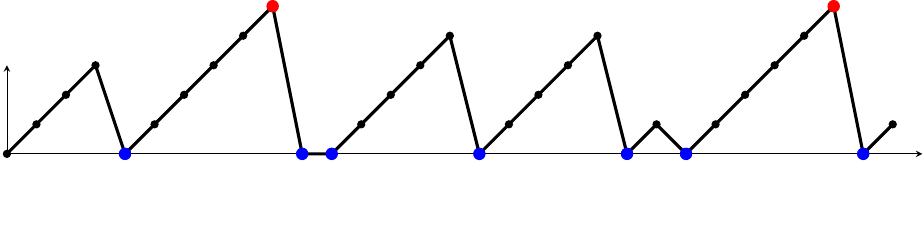}
\setlength\abovecaptionskip{-6mm}
\setlength\belowcaptionskip{0pt}
\caption{A Moran walk is a random walk which makes a jump $+1$ with probability~$p$, and a reset (a jump to 0) with probability $1-p$. Above, one sees such a walk of length $n=30$. Its final altitude is $Y_n=1$, the height is $H_n=5$ (reached twice, in red), having 7 resets (the 7 blue dots). 
In this article, we tackle the enumeration and asymptotics of such paths (and of generalizations involving more general step sets and higher dimension).
We also prove that this simple model of walks leads to some noteworthy nontrivial asymptotic behavior of their height $H_n$.}\label{Fig:Moran}
\end{figure}
\pagebreak

\section{Walks with resets: final altitude and height}\label{Sec2}
We consider walks with steps in $\mathcal S$ (where $\mathcal S$ is a nonempty finite subset of $\Z$),
which can additionally have a reset at any altitude. 
That is, we have the following process on $\Z$:

\vspace{-5mm}

\begin{align*}
Y_0&=0
\\
Y_{n+1}&= \left\{
 \begin{array}{ll}
 Y_n+k, & \hbox{ with probability } p_k \text{\qquad (for each $k\in \Z$, with $p_k:=0$ if $k\not \in \mathcal S$)},
\\
\\
 0, & \hbox{ with probability } q \text{\qquad (with $q+\sum_{k \in \mathcal S} p_k=1$)}.
 \end{array}
 \right.
\end{align*}
(So if $Y_n=0$ we have $Y_{n+1}=0$ with probability $p_0+q$.)

\noindent Thus, $Y_n$ is the altitude of the process after $n$ steps
and $H_n:=\max(Y_0,\dots,Y_n)$ is its height.
It is convenient to encode the steps and their probabilities by the Laurent polynomial 
\begin{equation}\label{defP}
P(u):=\sum_{k=c}^d p_k u^k \text{\qquad (with $c:=\min {\mathcal S}$ and $d:=\max {\mathcal S}$)}.
\end{equation}
We assume $0<q<1$ to avoid degenerate cases. 
We do not require that $c<0$ or $d>0$.
Of course, if $c\geq 0$, the walk will live by design in $\N$ (it is e.g.~the case for Moran walks of Figure~\ref{Fig:Moran}).
In Section~\ref{Sec2.1}, we determine the distribution of the final altitude 
(as illustrated in Figure~\ref{Fig:alt} for different families of steps) and we investigate the height in Section~\ref{sec:height}.

\begin{figure}[h]
\def\myL{.36\textwidth}
\includegraphics[width=\myL]{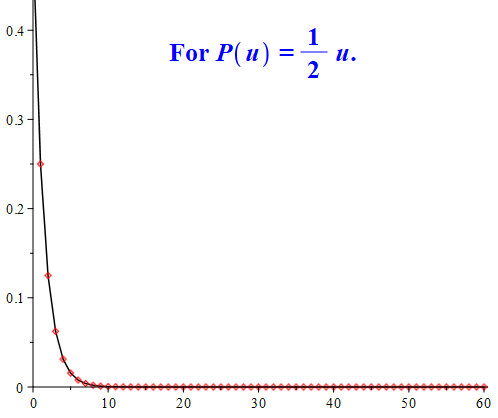} \hspace{1cm}
\includegraphics[width=\myL]{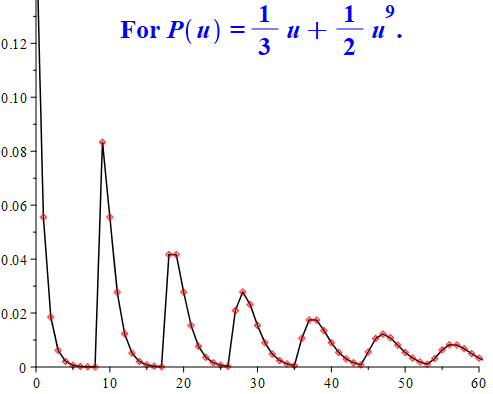}\\[4mm]
\includegraphics[width=\myL]{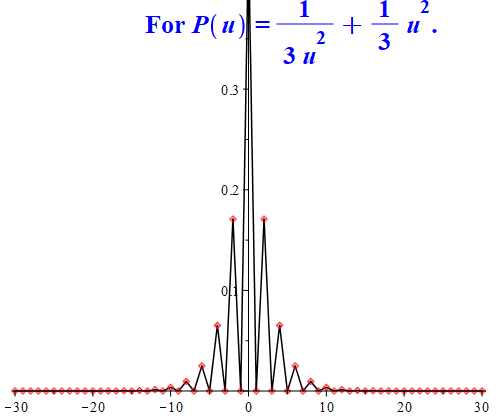} \hspace{1cm}
\includegraphics[width=\myL]{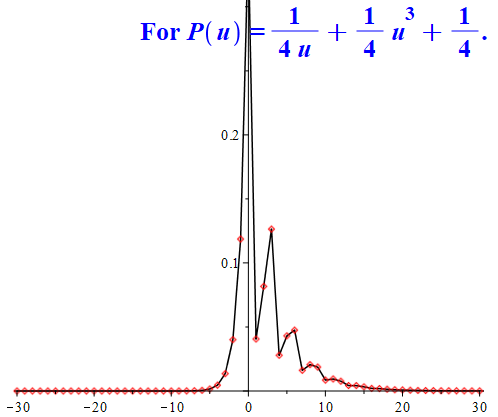}
\setlength\belowcaptionskip{0mm}
\caption{Plot of $\Pr(Y_n=k)$, the distribution of the altitudes of walks with resets, for $n=100$ and different $P(u)$.
It has its support in the $\N$-linear combinations of steps from $\mathcal S$.
The final altitude is of order $O(1)$ and the probability to end at higher altitudes decreases exponentially fast (see 
Theorem~\ref{Theorem1} for closed-form expressions of the mean and the distribution).}
\label{Fig:alt}
\end{figure}

\pagebreak

\subsection{Final altitude \texorpdfstring{$Y_n$}{Yn}}\label{Sec2.1}
Let us start with a simple result which paves the way for the more subtle 
generating function manipulations for the height that we tackle later in Section~\ref{sec:height}.

We use the classical convenient notations:
\begin{itemize}
\item $[z^n] G(z)$ stands for the coefficient of $z^n$ in the power series $G(z)$, 
\item $\partial_u^jF(z, 1)$ is the $j$-th derivative of $F(z, u)$ with respect to~$u$, evaluated at $u = 1$.
\end{itemize}
\begin{theorem}[Final altitude at finite time]\label{Theorem1}
The final altitude of walks with resets 
follows a discrete law with probability generating function\vspace{-1mm}
\begin{equation}\label{finalalt}
F(z, u) =\sum_{n\ge 0} \E[u^{Y_n}]z^n= \frac{1 + qz/(1 -z)}{1 - z P(u)},\vspace{-1.1mm}
\end{equation}
where $P(u)$ is the Laurent polynomial encoding the allowed steps (a finite subset of $\Z$).
Equivalently, for $k\in \Z$, we have\vspace{-1.5mm}
\begin{equation}\label{Ynk}
\Pr(Y_n=k) = [u^k] P(u)^n + q [u^k] \sum_{j=0}^{n-1} P(u)^j.\vspace{-1.5mm}
\end{equation}
\noindent Let $\delta:=P'(1)$ be the drift\footnote{We recall that $P(1)=1-q$, so another convention could have been to call drift the quantity 
$P'(1)/(1-q)$, i.e.,~we would then \emph{condition} on having no reset (instead of considering walks without~reset, weighted by the initial model~\eqref{defP}).
This alternative convention does not simplify the subsequent formulas.}
of the walk without reset, and $V:=P''(1)$ its second factorial moment.
The mean and the variance of the final altitude of the walk with resets are~given~by
\begin{equation*}
\E[Y_n]=\delta/q + (1-q)^{n-1} (\delta-\delta/q),
\end{equation*}
\begin{equation*}
\Var[Y_n]= \frac { \left( V+\delta \right) q+\delta^2}{q^2}
+(1-q)^n \left(2\,\frac {\delta^2 n}{ ( q-1) q}-\frac {V+\delta}{q}\right)
-(1-q)^{2n} \frac{\delta^2}{q^2}.
\end{equation*}
For Moran walks (i.e.,~$P(u)=pu$ and $p=1-q$), the mean and the variance simplify to
\begin{equation*}
\E[Y_n]
=\frac{p}{q}\Big(1-p^{n}\Big)
\text{\quad and \quad} 
\Var[Y_n]=
\frac{p}{q^2}\Big(1-p^n\big(p^{n+1}+(1+2n)q\big)\Big).
\end{equation*}
\end{theorem}
\begin{proof}
The probability generating function can be written as\vspace{-1mm}
\begin{equation*}
F(z,u)=\sum_{n\ge 0}\left(\sum_{k \in \Z}^n\Pr(Y_n=k)u^k\right)z^n =\sum_{n\ge 0}f_n(u)z^n ,\vspace{-1mm}
\end{equation*}
where the $f_n(u)$'s are Laurent polynomials encoding the location of the walk at time $n$;
thus we have $f_{n+1}(u) = P(u) f_n(u) + qf_n(1)$, with $f_0(u) = 1$. Multiplying both sides of this recurrence by $z^{n+1}$, and summing over $n$, one gets 
\begin{equation*}
F(z, u)(1 - z P(u)) = 1 + qzF(z, 1).
\end{equation*}
As $F(z, 1) = 1/(1 - z)$, one obtains Formula~\eqref{finalalt}.
Note that the generating function can also be obtained by using a regular expression encoding these walks (by factorizing the walk in factors ending by a reset): 
$({\mathcal S}^* q)^* (\mathcal S)^*$,
which translates to
\begin{equation*}
F(z,u) = \frac{1}{1-qz \frac{1}{1-zP(1)}} \frac{1}{1-zP(u)},
\end{equation*}
where the occurrences of $P(1)$ and $P(u)$ reflect that only the altitudes after the last reset
contribute to the final altitude of the full walk.
Using $P(1)=1-q$, we get Formula~\eqref{finalalt}.

The mean of $Y_n$ is then obtained via $\mu_n := \E[Y_n] = [z^n]\partial_uF(z, 1)$,
while its variance is obtained via a second-order derivative:
$\Var[Y_n] = [z^n]\partial_u^2F(z, 1) +\mu_n- \mu_n^2$.
\end{proof}
\pagebreak

We can now establish the corresponding limit distribution.

\begin{theorem}[Final altitude: asymptotics]\label{Theorem1bis}
Consider walks with $0\not \in \S$,  $\gcd \S=1$, and ${d=\max \S> 0}$ (these three constraints bring no loss of generality\footnote{
There is no loss of generality. Indeed, if the walk as a periodic support (i.e.,~if $\gcd(\S)=g$ with $g>1$)
we rescale (without loss of generality) the step set $\S$ by dividing each step by $g$.
Now, if $\max \S<0$, then we multiply each step by $-1$. 
Last, if $0 \in \S$ we consider instead the equivalent model $\S:=\S \setminus \{0\}$ and $q:=q+p_0$.}).
Therefore the support of the walk is either $\Z$ (with all altitudes being reachable), 
or $\N$ (with a finite set of altitudes impossible to reach, known as 
the unreachable set in the coin-exchange problem of Frobenius).
The final altitude of these walks with resets behaves asymptotically according to these two cases.
\bgroup
\setitemize{labelindent=8mm,labelsep=2mm,leftmargin=*}
\begin{itemize}[topsep=3pt]
\item[\em a)] For walks with $\min {\mathcal S}\geq 0$, we have for $k\in\N$ (not in the Frobenius unreachable set):\vspace{-1mm}
\begin{equation}\label{crude}
q \cdot (\min_{i\in\S} p_i)^k   \leq   \lim_n \Pr(Y_n=k) \leq  q \cdot (\max_{i\in\S} p_i)^{k/d}.
\end{equation}
In particular, for Moran walks, we have  $\Pr(Y_n=k)= qp^k$ for $0\leq k<n$ and  $\Pr(Y_n=n)= p^n$
so $\lim Y_n=  \operatorname{Geom}(q)-1$.
\item[\em b)] For walks with $\min \S <0$ and $\max \S>0$, we have for $k \in \Z$:\vspace{-1mm}
\begin{equation*}
\Pr(Y_n=k) =  q W_k(1-q)  +    (1-q)  \frac{1}{\tau^{k+1}}  \frac{1}{\sqrt{2\pi n P''(\tau)}} + O\left(\frac{1}{n}\right).
\end{equation*}
\end{itemize}
\egroup
Moreover, both in Case a) and in Case b), 
$\Pr(Y_n=k)$ has a geometric decay for large $k$.
\end{theorem}
\begin{proof}
In Case a),  we have $\min \S\geq 1$; 
the definition of $P(u)$ in~\eqref{defP} then entails $[u^k] P(u)^j=0$ for large $j$.
The limit of Equation~\eqref{Ynk}  thus gives 
\begin{equation}
\lim_{n\rightarrow +\infty}  \Pr(Y_n=k) =  q [u^k] \sum_{j=0}^{k} P(u)^j. 
\end{equation}
In particular, when it is not $0$, this quantity is lower bounded by $q \cdot (\min_{i\in\S} p_i)^k$ and upper bounded by $q \cdot (\max_{i\in\S} p_i)^{k/d}$,
and therefore decreases geometrically.

\def\PP{\widetilde P}
In Case b), the proof is more complicated and will recycle ingredients
of the asymptotics of walks without reset.
To this aim, first set $\PP(u) := P(u)/P(1)$, i.e.,~the step set\linebreak probabilities are renormalized to have global mass $\PP(1)=1$.
Let $W_k(z)$ be the probability generating function of walks without reset, i.e.,~$W_k(z)=[u^k] \frac{1}{1\sm z\PP(u)}= \sum_{n\geq 0} w_{n,k} z^n$.
We then rewrite Equation~\eqref{Ynk} as
\bgroup
\begin{align}
\Pr(Y_n=k) &= P(1)^n [u^k] \PP(u)^n + q [u^k] \sum_{j=0}^{n-1} P(1)^j \PP(u)^j \nonumber\\[-1.5mm]
 &= (1-q) P(1)^n w_{n,k} + q \sum_{j=0}^{n} P(1)^j w_{j,k} \nonumber\\[-1.5mm]
&= (1-q) P(1)^n w_{n,k} + q P(1)^{n} [z^n] \frac{1}{1-z/P(1)} W_k(z). \label{Wkz}
\end{align}
\egroup

If $\min {\mathcal S}<0$ and $\max {\mathcal S}>0$, then there is a unique real $\tau>0$ such that $\PP'(\tau)=0$.
It is proven in~\cite{BanderierFlajolet2002} that $\rho=1/\PP(\tau)$ is the radius of convergence of $W_k(z)$ and that 
$w_{n,k}\sim \tau ^{-k} C \PP(\tau)^n/\sqrt{2\pi n}$, where $C:=\frac{1}{\tau} \sqrt{\PP(\tau)/\PP''(\tau)}$. %(This follows by a saddle-point method applied to the Cauchy integral associated to $[u^k]1/(1-z\\P(u))$.)
\pagebreak

Note that, as we have a probability generating function, we have $\rho=\widetilde  P(\tau)=1$.
The asymptotics of~\eqref{Wkz} then follows by singularity analysis, as  $1/(1-z/P(1))$ is singular at $z=P(1)=1-q$, that is,  
before $W_k(z)$ which is singular at $z=1$:
\begin{equation}
\Pr(Y_n=k) =   q W_k(1-q) + (1-q) \tau^{-k} C \frac{ P(\tau)^n}{\sqrt{2\pi n}} + O\left(\frac{1}{n}\right).
\end{equation}
Note that Formulas (10) and (11) in \cite[Theorem 1]{BanderierFlajolet2002} 
give a closed form for $W_k(z)$. It implies in particular 
\begin{equation} 0< W_k(1-q) < (1-q) (c+d) C_1/C_2^{|k|+1},\end{equation} 
where $C_1>0$ and $C_2>1$ are constants independent of $k$; 
thus  $W_k(1-q)$ decays geometrically for $k\rightarrow \pm \infty$. 
This concludes our analysis of Case b) and gives the theorem.
\end{proof}
These limiting behaviors are thus in sharp contrast with the asymptotic behavior of the final altitude 
of walks on $\Z$ with no resets, which is $\delta n \pm O(\sqrt n)$, 
with fluctuations given by a continuous distribution (Rayleigh or Gaussian; see~\cite{BanderierFlajolet2002}). 

\subsection{The height \texorpdfstring{$H_n$}{Hn}}\label{sec:height}
In order to study the height of these walks with resets, one considers the subset of them made of walks conditioned to have a height smaller than~$h$.
We want to obtain an explicit formula for their generating function 
\begin{equation*} F^{\le h}(z, u) := \sum_{n=0}^{+\infty}\E\Big(u^{Y_n}{\indi}_{\{Y_1\le h,Y_2\le h,\dots,Y_n\le h\}}\Big)z^n.\end{equation*}

If these walks are generated by a step set $\mathcal S$ having only positive jumps, 
a natural but naive approach to enumerate them would be 
to create a deterministic finite automaton (a finite discrete Markov chain) with $h$ states
encoding the possible altitudes of the process.
It leads to a system of linear equations which would allow us to get the corresponding
rational generating function. However, this approach to obtain the generating function (given $h$ and the transition probabilities) suffers from three drawbacks: 
\newlist{myitemize}{itemize}{3}
\setlist[myitemize,1]{label=\textbullet,leftmargin=28pt}
\begin{myitemize}
\item it would be of complexity $h^3$ (computing determinants of $h \times h$ matrices), 
\item it would be a case-by-case approach
(new computations are needed for each~$h$),
\item it would fail if the step set $\mathcal S$ has some negative steps (then the support of the walk\linebreak is $[-\infty,+h]$, and thus one would need an automaton with an infinite number~of~states).
\end{myitemize}

So, we prefer here to use a more efficient approach,
which relies on a powerful method (namely, the kernel method~\cite{BanderierNicodeme2010}):
the complexity to obtain a closed-form formula for $F^{\le h}(z, u)$
then drops\footnote{The PhD thesis of Louis Dumont~\cite{Dumont2016} compares the cost of different methods 
to compute the \textit{coefficients} of such generating functions (which can be related to diagonals of rational functions); the full analysis has to take into account the space and time complexities,
and some precomputation steps, of cost of course higher than $O(1)$, 
but in all cases it is more efficient than a Markov chain approach
(see however Bacher~\cite{Bacher2023} for a clever use of a transfer matrix point of view).}
 from $O(h^3)$ to $O(1)$ for any finite step set $\S \subset \Z$\,! 
This leads to the following theorem.
\pagebreak

\begin{theorem}\label{ThFh}
Let $F^{\le h}(z, u)$ be the probability generating function of walks on $\Z$ of height~$\leq h$ with resets,
where the length and the final altitude of the walks are respectively encoded by the exponents of $z$ and $u$. 
Let $P(u)$ encode the allowed jumps as in~\eqref{defP}.~One~has 
\begin{align}\label{generic1}
F^{\le h}(z, u) &= \sum_{n=0}^{+\infty}\E\Big(u^{Y_n}{\indi}_{\{Y_1\le h,Y_2\le h,\dots,Y_n\le h\}}\Big)z^n 
= \frac{W^{\leq h}(z,u)}{1-zq W^{\leq h}(z,1)} ,
\end{align}
\vspace{-1.2mm}
\noindent where \vspace{-2mm}
\begin{align}\label{eqM}
W^{\leq h}(z,u) &:=\frac{\displaystyle{1 -\sum_{i=1}^d \left(\frac{u}{u_i}\right)^{h+1} \prod_{1\leq j \leq d, j\neq i} \frac{ u_j-u}{u_j-u_i} } }{\displaystyle{1-zP(u)}}
\end{align}
is the generating function of walks of height~$\leq h$ without reset, 
and where $u_1,\dots, u_d$ are the roots of $1-zP(u)=0$ such that $\lim_{z\rightarrow 0} |u_i(z)|=+\infty$.
\end{theorem}
\begin{remark}[A \emph{rational} simplification]
These generating functions are algebraic, as they rationally depends on the roots $u_i(z)$, which are themselves algebraic functions.
Now, when the step set~$\mathcal S$ has only \emph{positive} steps, $W^{\leq h}$ is a polynomial and $F^{\leq h}$ simplifies to a rational function
(despite the fact that their closed forms \eqref{eqM} and \eqref{generic1} involve algebraic functions!). 
This simplification can be seen either by the automaton point of view and the Kleene theorem,
or by using the Vieta formulas on Newton sums (as, when one has only positive jumps, the $u_i$'s are then \textit{all} the roots of the kernel $1-zP(u)$).
For example, for $P(u)=u/3+u^2/2$ and $h=3$, we have 
\begin{equation}
u_1(z)=\frac{-z + \sqrt{z^2 + 18z}}{3 z}  \text{ \qquad and  \qquad } 
u_2(z)=\frac{-z - \sqrt{z^2 + 18z}}{3z}
\end{equation}
(the Vieta formulas are here: $u_1(z)+u_2(z)= -2/3$ and $u_1(z) u_2(z)=-2/z$); 
then, the quotient \eqref{generic1} involving these algebraic functions $u_1$ and $u_2$ simplifies, leading to
\begin{align*}
W^{\leq 3}(z,u)&= 
\frac{1}{1-zP(u)}   \left(1 -\left(\frac{u}{u_1(z)}\right)^4 \frac{ u_2(z)-u}{u_2(z)-u_1(z)} - \left(\frac{u}{u_2(z)}\right)^4 \frac{ u_1(z)-u}{u_1(z)-u_2(z)} \right) \\
&= 1+z \left( {\frac {u^2 }{2}}+{\frac {u}{3}} \right) +z^2 \left( {\frac {u^3}{3}}+{\frac {u^2 }{9}}
\right) +\frac {z^3 u^3}{27}, \\
F^{\leq 3}(z,u) &=\frac{ \left( 1+z \left( {\frac {u^2 }{2}}+{\frac {u}{3}} \right) +z^2 \left( {\frac {u^3}{3}}+{\frac {u^2 }{9}} \right) +\frac {z^3 u^3}{27} \right) }
{1-zq \left( 1+{\frac {5z}{6}}+{\frac {4{z}^2 }{9}}+{\frac {z^3}{27}} \right) }.
\end{align*}
\end{remark}

\begin{proof}[Proof of Theorem~\ref{ThFh}]
The probability generating function can be written as 
\begin{equation*}
F^{\le h}(z,u)
 = \sum_{n\ge0}f_n^{\le h}(u)z^n =\sum_{k=0}^h F^{\le h}_k(z)u^k,
\end{equation*}
\noindent
where $f_n^{\le h}(u)$ encodes the possible values of $Y_n$ (constrained to be bounded by $h$ over the full process), and where
\begin{equation*}
F^{\le h}_k(z)=\sum_{n=0}^{+\infty} f_{n,k}^{\leq h} z^n = \sum_{n=0}^{+\infty}\Pr\Big(Y_1\le h,\,Y_2\le h,\dots,Y_{n-1}\le h,\,Y_n=k\le h\Big)z^n
\end{equation*}
\noindent is the probability generating function of bounded walks ending at altitude $k$.
\pagebreak

The dynamics of the process then entails the recurrence
\begin{equation*}
f_{n+1}^{\le h}(u)=P(u) f_n^{\le h}(u)-\{u^{>h}\} P(u) f_{n,h}^{\le h} u^h +qf_n^{\le h}(1),
\end{equation*}
where $\{u^{>h}\}$ extracts monomials having a degree in $u$ strictly larger than $h$. 
This mimics that at time~$n+1$,
either, with probability $p_k$, we increase by $k$ the altitude of where we were at time $n$ (that is, we multiply by $u^k$, 
and this is allowed as long as the walk stays at some altitude $\leq h$, thus we removed here the cases corresponding to the walks which would reach an altitude $>h$ at time $n+1$); 
or, with probability $q$, we have a reset to altitude 0 (i.e.,~all the mass of the walks at any altitude $k$, corresponding to the coefficient of $u^k$, is sent back to $u^0$; this is thus captured by the substitution $u=1$). 

This directly translates to the functional equation
\begin{equation*}
F^{\le h}(z, u)=1+zP(u)F^{\le h}(z,u)- \sum_{k=0}^{d-1} F_{h-k}^{\le h}(z)u^{h-k} \left(z \sum_{j=k+1}^d p_j u^j \right) +zqF^{\le h}(z,1).
\end{equation*}

Setting $q=0$, we get the functional equation for 
the generating function $W^{\le h}$ of walks of height~$\leq h$ without reset:
\begin{equation}\label{eq:W}
W^{\le h}(z, u)=1+zP(u)W^{\le h}(z,u)- \sum_{k=0}^{d-1}W_{h-k}^{\le h}(z)u^{h-k} \left(z \sum_{j=k+1}^d p_j u^j \right) .
\end{equation}

Of course, the factorization of walks with resets into $({\mathcal S}^* q)^* (\mathcal S)^*$ 
entails $F^{\leq h}(z,u)=\operatorname{Seq}(W^{\leq h}(z,1)q) W^{\leq h}(z,u)$,
which is Formula~\eqref{generic1}.
So if we find a closed form for $W^{\le h}$, we are happy as this also solves the initial problem for $F^{\le h}$.
Now, on the right-hand side of~\eqref{eq:W}, the sum for $k$ from $0$ to $d-1$ is a polynomial in $u$, which we conveniently~rewrite~as
\begin{equation}\label{eq:kern}
W^{\le h}(z, u)(1 - zP(u)) = 1 - u^h \sum_{k=1}^{d} G_k(z) u^k.
\end{equation}

It is possible to solve such an equation via the kernel method:
the kernel is the factor $1-zP(u)$ in~\eqref{eq:kern}, and if one considers the equation on the variety 
defined by $1-zP(u)=0$, this brings additional equations which will allow us to get a closed form for $W^{\le h}(z, u)$.
First, observe that this kernel is a (Laurent) polynomial in $u$ of ``positive'' degree $d$. Then, from an analysis of its Newton polygon,
one gets that it has $d$ roots $u_1(z), \dots, u_d(z)$ such that $u_i(z)\approx z^{-1/d}$ for $z\sim0^+$
(the other roots being convergent at $z\sim 0^+$; see~\cite{BanderierFlajolet2002} for more on this issue).
Thus, setting $u=u_i(z)$ (for $i=1, \dots, d$) in the functional equation~\eqref{eq:kern} gives $d$ new equations. Some care is required in this step: we have to check that one does not create series involving an infinite number of monomials with negative exponents\footnote{Let $R$ be the ring of series $\sum_{n \in \Z} a_n z^n$. The Cauchy product of two series in $R$ is well defined only with some additional convergence conditions, 
and, even if we restrict ourselves to series for which the product is well defined, we have to take care to the fact that they do not form an integral ring: indeed, we have many divisors of zero (e.g.~for~$S(z):=\sum_{n \in Z} z^n$, we have $zS=S$ and thus $(z-1) S=0$). 
Most algebraic manipulations in this ring, if they are temporarily handling quantities which are not in the subring of power series (or Laurent/Puiseux/Fourier series),
would lead to invalid identities in~${\mathbb C}[[z]]$.}.
\pagebreak 

In fact, in our case, the substitution $u=u_i$ is legitimate as $W^{\leq h}(z,u_i)$ becomes a well-defined Puiseux series in $z$: this follows from the fact that the coefficients $f_n^{\leq h}(u)$ are (Laurent) polynomials with ``positive'' degree bounded by $h$ (and ``negative'' degree lower bounded by $-cn$),
so $f_n^{\leq h}(u_i(z))$ is a Puiseux series with exponents from $-h/d$ to~$+\infty$. Then, multiplying by $z^n$ and summing over $n$, 
only a finite number of summands contribute to each monomial of $W^{\leq h}(z,u_i)$, which is thus well defined.
Via these substitutions $u=u_i$, we obtain a linear system of $d$ equations (which only contains the $G_k$'s as unknowns). 
Then, by Cramer's rule, we get $G_k=\det(V_k)/\det(V)$, where
\begin{equation*}
V=
\begin{pmatrix}
u_1^{h+1} & u_1^{h+2} & \dots & {u_1}^{h+d}\\
u_2^{h+1} & u_2^{h+2} & \dots & {u_2}^{h+d}\\
\vdots & \vdots & \vdots & &\vdots \\
u_d^{h+1} & u_d^{h+2} & \dots & {u_d}^{h+d}
\end{pmatrix}
\text{\quad and \quad}
V_k=
\begin{pmatrix}
u_1^{h+1} & \dots & u_1^{h+k-1} & 1 & u_1^{h+k+1} & \dots & {u_1}^{h+d}\\
u_2^{h+1} & \dots & u_2^{h+k-1} & 1 & u_2^{h+k+1} & \dots & {u_2}^{h+d}\\
\vdots & \vdots & \vdots & &\vdots \\
u_d^{h+1} & \dots & u_d^{h+k-1} & 1& u_d^{h+k+1} & \dots & {u_d}^{h+d}
\end{pmatrix},
\end{equation*}
that is, $V_k$ is the matrix $V$ with its $k$-th column entries replaced by~$1$.
 Thus, as $V$ is a Vandermonde matrix, its determinant is
\begin{equation}
\det(V)=\left(\prod_{i=1}^d u_i^{h+1}\right) \prod_{1\leq i < j \leq d} (u_j-u_i).
\end{equation}
Now, to compute $\det(V_k)$, one first proves that 
\begin{align}\label{DET}
\Delta=\det
\begin{pmatrix}
u_1^{1} & \dots & u_1^{k-1} & 1 & u_1^{k+1} & \dots & {u_1}^{d}\\
u_2^{1} & \dots & u_2^{k-1} & 1 & u_2^{k+1} & \dots & {u_2}^{d}\\
\vdots & \vdots & \vdots & &\vdots \\
u_d^{1} & \dots & u_d^{k-1} & 1 & u_d^{k+1} & \dots & {u_d}^{d}\\
\end{pmatrix}
= e_{d-k} (u_1,\dots,u_d) \prod_{1\leq i < j \leq d} (u_j-u_i),
\end{align}
where we used the classical notation for the elementary symmetric polynomials:
\begin{equation}\label{ek} e_k( x_1,\dots,x_d):=[t^k] \prod_{i=1}^d (1+t x_i),\end{equation}
e.g.,
$e_3(x_1,\dots,x_5)=x_1 x_2 x_3 +x_1 x_2 x_4 +x_1 x_2 x_5 +x_1 x_3 x_4 +x_1 x_3 x_5 +x_1 x_4 x_5 +x_2 x_3 x_4 +x_2 x_3 x_5 +x_2 x_4 x_5 +x_3 x_4 x_5$.
Formula~\eqref{DET} follows from 2 facts:
\begin{itemize} 
\item If $u_i=u_j$, then two rows of $V_k$ are equal and thus the determinant is 0; this explains the Vandermonde product
$\Pi:= \prod_{1\leq i < j \leq d} (u_j-u_i)$
on the right-hand side of Formula~\eqref{DET}.
\item Now writing the determinant as a sum over the $d!$ permutations of the entries gives a sum of monomials, each of total degree $(1+2+...+d)-k$ in the $u_i$'s. $\Pi$ being of total degree $\binom{d}{2}=d (d-1)/2$, it implies that $\Delta/\Pi$ is a polynomial which is symmetric and homogeneous of total degree $d-k$. Up to a constant factor (determined to be 1, by comparing any monomial), this polynomial has to be $e_{d-k}$, which captures exactly the missing $u_i$'s in each of the $d!$ summands.
\end{itemize}

Then, performing a Laplace expansion of $\det(V_k)$ on its $k$-th column and using Formula~\eqref{DET}, 
one gets (after simplification in the Cramer formula):
\begin{equation}
G_k(z)=\sum_{\ell=1}^d u_\ell^{-h-1} (-1)^{k+d} e_{d-k}(u_1,\dots,u_d)_{|u_\ell=0} \prod_{\substack{1\leq j \leq d\\ j\neq \ell}} \frac{1}{u_\ell-u_j}.
\end{equation}

\pagebreak

Now, using 
$\sum_{k=0}^d (-1)^{d-k} e_{d-k}(u_1,\dots,u_d) u^{k} = \prod_{i=1}^d (u-u_i)$ (which is equivalent to
the definition~\eqref{ek}), and regrouping the powers $u_k^{-h-1}$, we get
\begin{equation}\label{Interpolation}
\sum_{k=1}^d G_k(z) u^{k-1}= \sum_{k=1}^d u_k^{-h-1} \prod_{1\leq j \leq d, j\neq k} \frac{ u_j-u}{u_j-u_k}.
\end{equation}
Combining Equations~\eqref{Interpolation} and~\eqref{eq:W}, we get Formula~\eqref{eqM} for $W^{\leq h}(z,u)$, and thus the closed form for $F^{\le h}(z, u)$.
\end{proof}

\begin{remark}[Link with Lagrange interpolation]
As we know the evaluation of the right-hand side of~\eqref{eq:kern} in each of the $u_k$, 
Formula~\eqref{Interpolation} is also equivalent to the Lagrange interpolation formula (which we thus reproved \textit{en passant}).
Moreover, this Lagrange interpolation approach offers a nice advantage: it is circumventing the fact that the factorization argument 
used to get the closed forms for the generating functions in~\cite{BousquetMelouPetkovsek2000,BanderierFlajolet2002} 
works only if the walks start at altitude 0.
\end{remark}

\smallskip

Now, if we go back to Moran walks (i.e.,~for $P(u)=pu$; see Figure~\ref{Fig:Moran}), 
the generating function simplifies to the following noteworthy shape. 

\begin{corollary}\label{PropFh}
The probability generating function of Moran walks of height~$\leq h$ is 
\begin{align}\label{generic}
F^{\le h}(z, u) %= \sum_{n=0}^{+\infty}\E\Big(u^{Y_n}{\indi}_{\{Y_1\leh,Y_2\le h,\dots,Y_n\le h\}}\Big)z^n 
=\frac{(1 - pz)(1 -(pzu)^{h+1})}{(1 -puz) (1 - z + (pz)^{h+1}zq)},
\end{align}
where, in the power series, the length and the final altitude of the walks are respectively encoded by the exponents of $z$ and $u$.
Accordingly, 
\begin{align}\label{FPHn}
\hspace{-8mm} \Pr(H_n\leq h)&=[z^n] F^{\leq h}(z, 1)=[z^n]\frac{1 -(pz)^{h+1}}{1 - z + (pz)^{h+1}zq}\\
%&=\sum_{k=0}^{\left\lfloor \frac{n}{h+1} \right\rfloor} \binom{n-k(h+1)}{k} (-qp^{h+1})^k
%-p^{h+1} \sum_{k=0}^{\left\lfloor \frac{n}{h+1} \right\rfloor-1} \binom{n-(k+1)(h+1)}{k} (-qp^{h+1})^k\\
&=\sum_{k=0}^{\left\lfloor \frac{n}{h+1} \right\rfloor} 
 (-qp^{h+1})^k \left( \binom{n-k(h+1)}{k}-p^{h+1} \binom{n-(k+1)(h+1)}{k} \right),\qquad \label{PippengerGeneralized}
\end{align}
with the convention that $\binom{m}{k}=0$ if $m<0$.
\end{corollary}
\begin{proof}The closed form~\eqref{PippengerGeneralized} is obtained via the power series expansion $1/(1-T)=\sum T^j$ by applying the binomial theorem to each term $T^j$,
with $T=z + (pz)^{h+1}zq$.\end{proof}

The binomial sum~\eqref{PippengerGeneralized} generalizes 
a formula obtained (for $p=1/2$) by Pippenger in~\cite{Pippenger2002}.
Therein, it is derived by an inclusion-exclusion principle (guided by the combinatorics of the carry propagation in binary words);
for his problem, the generating function, and thus the corresponding binomial sum, are a little bit simpler than~\eqref{FPHn} and~\eqref{PippengerGeneralized},
and are then used to perform some real analysis for the asymptotics of the expected length.

In our case, equipped with this explicit expression for the probability generating function of Moran walks of bounded height, 
we can now tackle the question of the asymptotic distribution of this extremal parameter. %, with a complex analysis method.

\pagebreak
\section{Asymptotic height of Moran walks}\label{Sec3}
In this section, we establish a local limit law for the distribution of the height of Moran walks. 
One noteworthy consequence of the generating function explicit formula that we get in the previous section
is that it allows us to have very efficient computations and simulations of the process at time $n$, for large $n$, 
as stressed by the following remark.\smallskip

\begin{remark}[Fast computation scheme for any given $n$ and $h$]
One does not need to run the process for $n$ steps
to have the exact distribution of $H_n$. Indeed,
using the rational generating function from Corollary~\ref{PropFh},
for any $p$, $h$, and $n$, it is possible to get the exact value 
of\/ $\Pr\left(H_n = h\right) = [z^n] \left(F^{\le h}(z, 1)-F^{\le h-1}(z, 1)\right)$
in time $O(\ln(n))$ via binary exponentiation. 
\end{remark}
\smallskip

This allows
us to plot the distribution $H_n$, for quite large values of~$n$ (as an example, see Figure~\ref{plotHn}). 
Note that for our other generating functions, which are algebraic, there exists a fast algorithm of cost $\sqrt{n} \ln(n)$ to compute their $n$-th coefficient 
(this algorithm works more generally for all D-finite functions). This algorithm due to the brothers Chudnovsky is e.g.~implemented in the Maple computer algebra system
via the package Gfun; see~\cite{SalvyZimmermann1994}

\bgroup
\begin{figure}[h]
\includegraphics[width=.445\textwidth]{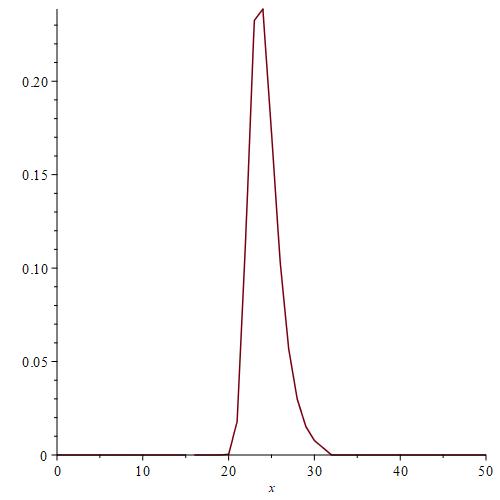}
\includegraphics[width=.445\textwidth]{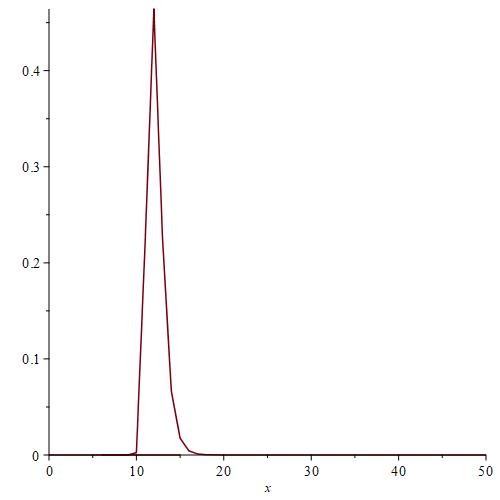}
\caption{The distribution of $H_n$, for $n=2^{25}$ (for $p=1/2$ on the left and $p=1/4$ on the right). One observes a sharp concentration around the height $25$ for $p=1/2$ and $12.5$ for $p=1/4$, suggesting a logarithmic link in base~$1/p$ between $n$ and $H_n$. We prove and refine this claim in~the~next~pages.}
\label{plotHn}
\end{figure}
\egroup

\vspace{-12cm}
\begin{picture}(200,200) 
\put(140,100) {$\Pr(H_n=h)$}
\put(140,80) {(for $p=\frac{1}{2}$)}
\put(305,100) {$\Pr(H_n=h)$}
\put(305,80) {(for $p=\frac{1}{4}$)}
\end{picture}
\vspace{5cm}

\subsection{Localization of the dominant singularity}

As $\Fh(z,1)$ (as given by Equation~\eqref{generic}) is a rational function, all its singularities are poles.
The asymptotic behavior of the coefficients of $\Fh(z,1)$ is governed by the closest pole(s) to zero 
(also called ``dominant singularities'' of $\Fh$).
A natural candidate for being such a dominant singularity of $\Fh(z,1)$ 
would be $z=1/p$, but it is in fact a removable singularity, as one has (e.g.~via L'H\^opital's rule)
$\Fh(1/p,1) = \frac{p (h+1)}{2p-1-q h}$. 
Thus, we can focus on the other roots of the denominator $D(z)$ of $\Fh(z,1)$.
\pagebreak
\begin{lemma}[Localization of the singularities of $\Fh$]\label{Lemma1}
For $p \in (0,1)$, the $h+2$ roots $z_1(h),\dots,z_{h+2}(h)$ of $D(z)=1-z+qp^{h+1}z^{h+2}$ are such that we have for $h$ large enough:
\bgroup
\begin{itemize}
\item[\em(i)] $z_1(h)$ is the unique root strictly between 1 and $1/p$;
\item[\em(ii)] $z_2(h) = 1/p$ is the unique root of modulus $1/p$;
\item[\em(iii)] the remaining $h$ roots $z_3(h),\dots, z_{h+2}(h)$ are all of modulus $>1/p$, and arbitrarily close (in modulus) to $1/p$ (for $h\rightarrow +\infty$);
\item[\em(iv)] all the roots are simple.
\end{itemize}
\egroup
Accordingly, $z_1(h)$ is the dominant singularity of $\Fh(z,1)$.
\end{lemma}
\begin{proof}Let $z_*(h)$ be the unique positive zero of $ D'(z)=-1+(h+2)qp^{h+1}z^{h+1}$
 given by
\bgroup
\setlength{\abovedisplayskip}{1.5pt}
\setlength{\belowdisplayskip}{3.5pt}
 \begin{equation*}
z_*(h)=\frac{1}{p}\left(\frac{1}{q(h+2)}\right)^{\frac{1}{h+1}}.
\end{equation*}
\egroup
As $z_*(h)$ tends to $\frac{1}{p}$ from the left, we thus have $0<z_*(h)<1/p$ for $h$ large enough.
Moreover, $D(z)$ is decreasing for all $z$ in the interval $ [0,z_*(h)]$ and increasing in the interval $ [z_*(h), + \infty]$. 
As $D(1/p)=0$, one thus has $D(z_*(h))<0$. And since $D(1)>0$, the intermediate value theorem implies 
the existence of (at least) one zero of $D$ between $1$ and~$z_*(h)$.
Combined with the (non)decreasing properties of $D$, this entails the unicity of this zero; let us call it $z_1(h)$.
Then, Pringsheim's theorem (see e.g.~\cite{FlajoletSedgewick2009}) asserts that $\Fh$  has a real positive dominant singularity which is thus $z_1(h)$, the first real positive zero of~$D$.
As $\Fh(z)$ is a probability generating function, all its singularities are of modulus~${\geq 1}$.
So we have  ${1<z_1(h)< z_*(h) <1/p}$ and thus proved (i).

We now prove (ii). The fact that $z_2(h)=1/p$ is a root follows from $1-1/p+q/p=0$. 
Is there any other root of the same modulus? If $z= \exp(i\theta)/p$ (with $\theta\in [0,2\pi]$) would be a root of $D(z)$,
then this would imply $p = \exp(i \theta) - q \exp(i (h+2) \theta)$.
By the reverse triangle inequality 
$\Big| |x| - |y| \Big| \leq |x-y|$ (with equality only if $xy=0$ or $x/y\in \R^+$),
this would entail $\theta=0$. 

To prove (iii), we use the following version of Rouch\'e's theorem:
if $|D-g| < |g|$ on the boundary of a disk $\mathcal D$, then $D$ and $g$ have the same number of roots inside $\mathcal D$.
We can apply this theorem to $D$ with $ g(z):=1-z$, for the disk ${\mathcal D}(0,\frac{1-\epsilon}{p})$:
on its boundary, one indeed has $|D(z)-g(z)|=\frac{q}{p}|pz|^{h+2} \leq \frac{q}{p}|1-\epsilon|^{h+2} <\frac{q}{p} |1-\epsilon|^{2/q} <  \frac{q-\epsilon}{p} \leq |g(z)|$,
where the first strict inequality holds for $h \geq 2/q$ and the next strict inequality holds for any small enough $\epsilon$ (independently of $h$), as we have then 
$\frac{ \ln(1-\epsilon/q)}{\ln(1-\epsilon)}<2/q$.
As the constraint on $h$ is independent of $\epsilon$, letting $\epsilon \rightarrow 0$, we infer that $D$ has only one root strictly inside ${\mathcal D}(0,\frac{1}{p})$.

Now we can also apply this theorem to $D$ with $g(z):=1+z^{h+2}$:
on the boundary of the disk ${\mathcal D}(0,\frac{1+\epsilon}{p})$, one indeed has, for $h$ large enough (depending on $\epsilon$), 
\bgroup
\setlength{\abovedisplayskip}{1.5pt}
\setlength{\belowdisplayskip}{3.5pt}
\begin{equation*}
|D(z)-g(z)|\le\left(\frac{1+\varepsilon}{p}\right)^{h+2}\left(1-qp^{h+1}\right)+\frac{1+\varepsilon}{p}<\left(\frac{1+\varepsilon}{p}\right)^{h+2}-1\leq |g(z)|,
\end{equation*}
\egroup
where the  last $-1$ is just a crude bound of the term $-\frac{q}{p} (1+\varepsilon)^{h+2}+\frac{1+\varepsilon}{p}$
which converges to $-\infty$  for $h \rightarrow +\infty$. 
So $D$, like $g$, has $h+2$ roots inside this disk.
\\
To prove (iv), note that the equation $D(z)=D'(z)=0$ is forcing $z=1+\frac{1}{h+1}$, but $D'(1+\frac{1}{h+1}) \rightarrow -1$ for $h \rightarrow +\infty$, therefore all the zeros are simple for $h$ large enough. 
\end{proof}

See Figure~\ref{roots} on page~\pageref{roots} for an illustration of the location of the roots.

\subsection{Limit distribution of the height: the discrete Gumbel distribution}

The height distribution exhibits some a priori surprising asymptotic aspects, having a flavor 
of number theory/Diophantine approximation. 
Such phenomena, however, appear for a few other probabilistic processes 
where some statistics could have different asymptotic behaviors 
depending on some resonance between $\ln p$ and $\ln q$ (see~e.g.~Janson~\cite{Janson2012}
or Flajolet, Vall\'ee, and Roux~\cite{FlajoletRouxVallee2010} for some examples related to tries or binary search trees).
In our case, it appears that a resonance between $\ln p$ and $\ln n$ plays a role.

\begin{theorem}[Distribution of the height of Moran walks]\label{Th:height_distribution}
We have
\begin{equation}\label{PHn}
\Pr\left(H_n\le h \right)=\exp\left(-qn p^{h+1}\right) \left(1+\errortermLNthree\right),
\end{equation}
where the error term is uniform for $h \in [0,n]$.
Accordingly, $\Pr(H_n=h)$ is unimodal, with a peak at $h=h^*(n)$, the closest integer to $c^*(n) \frac{\ln(n)}{\ln(1/p)}$,
where $c^*(n):=1-\frac{\ln(\ln(1/p)/q^2)}{\ln(n)}$, and we have 
\begin{equation} \Pr(H_n=h^*(n)) \sim p^{p/q}-p^{1/q}. \end{equation}
%Accordingly, still uniformly in $h$, $H_n$ has the following local asymptotic behavior
%\begin{equation*}
%\Pr\left(H_n = h \right) = n q \ln\left(\frac{1}{p}\right) p^{h+1} \exp\left(-qn p^{h+1}\right) \left(1+\errortermLNthree\right).
%\end{equation*}
Moreover, the mass is sharply concentrated around 
$\frac{\ln n}{\ln(1/p)}$, as better seen by the following result, with a uniform error term in~$k$:
\begin{equation*}
\Pr\left(H_n\le \left\lfloor\frac{\ln n}{\ln(1/p)}\right\rfloor+k\right)=\exp\left(-q\alpha(n)p^{k+1}\right) \left(1+\errortermLNthree\right),
\end{equation*}
with $\alpha(n):=p^{-\{\frac{\ln n}{-\ln p}\}}$ (where $\{x\}$ stands for the fractional part of $x$,
and where $\lfloor x \rfloor$ stands for the floor function of $x$). [See Figure~\ref{plotHn} on page~\pageref{plotHn} for an illustration of the distribution of~$H_n$ and Figure~\ref{Fig_sawtooth} for the behavior of the function $\alpha(n)$.]
\end{theorem}
%\vspace{-7mm}
\begin{figure}[bh]
\includegraphics[width=.45\textwidth]{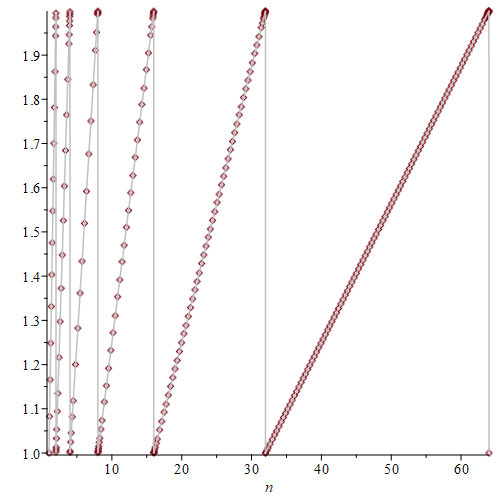}
\setlength\abovecaptionskip{2mm}
\caption{Plot of the function $\alpha(n)=p^{-\{\frac{\ln n}{-\ln p}\}}$ (for $p=1/2$), which occurs in 
the fluctuations of the height of Moran walks (as stated in Theorem~\ref{Th:height_distribution}).
The function $\alpha(n)$ is taking values in $[1,1/p)$ for integers $n\geq 1$. It has a sawtooth wave shape, with frequencies getting larger and larger (with peaks at powers of $1/p$).}\label{Fig_sawtooth}
\end{figure}

\begin{proof}
In the sequel, as the context is explicit, we simply denote by $z_1, \dots, z_{h+2}$ the zeros $z_1(h), \dots, z_{h+2}(h)$ of $D(z)=1-z+qp^{h+1}z^{h+2}$.
From Lemma~\ref{Lemma1}, for $h$ large enough, all these zeros $z_i$ are simple; the partial fraction decomposition of $1/D$ is then% \vspace{-1mm}
\begin{equation*} \frac{1}{D(z)}=\sum_{i=1}^{h+2}\frac{1}{D'(z_i)\left(z-z_i\right)}\end{equation*}
and as $D'(z_i)=-1+(h+2)(z_i-1)/{z_i}$, one thus gets
\begin{align*}
\Fh(z,1)&= \frac{1-(pz)^{h+1}}{D(z)} =\sum_{i=1}^{h+2}
\frac{1-(pz)^{h+1}}{D'(z_i)\left(z-z_i\right)} %=\sum_{i=1}^{h+2} \frac{(1-(qz)^{h+1})}{z_i-\left(z_i-1\right)(h+2)}\sum_{n=0}^{+\infty} (z/z_i)^n
\\
&=\sum_{i=1}^{h+2} \left( \frac{1}{z_i\sm\left(z_i\sm 1\right)(h\sp 2)}\left(\sum_{n=0}^{+\infty}z_i^{-n}z^n\right) \sm \frac{p^{h+1}}{z_i \sm \left(z_i \sm 1\right)(h\sp 2)}\sum_{n=h+1}^{+\infty}z_i^{-n+h+1}z^{n}\right)\\
&=\sum_{i=1}^{h+2} \left( \frac{1}{z_i\sm\left(z_i\sm 1\right)(h\sp 2)} \left(\sum_{n=0}^{h}z_i^{-n}z^n\right) + \frac{1\sm(pz_i)^{h+1}}{z_i \sm \left(z_i\sm 1\right)(h\sp 2)}\sum_{n=h+1}^{+\infty}z_i^{-n}z^n \right).
\end{align*}
It is combinatorially obvious that $\Pr\left(H_n\le h\right)=1$ for all $n\leq h$. 
So we now focus on $n> h$, for which we have, as $(pz_i)^{h+1}=\frac{z_i-1}{qz_i}$ and $1-\frac{z_i-1}{qz_i}=\frac{1-pz_i}{qz_i}$:
\begin{align}\Pr\left(H_n\le h\right)=[z^n] \Fh(z,1) &= \sum_{i=1}^{h+2} \frac{1-(pz_i)^{h+1}}{z_i-\left(z_i-1\right)(h+2)}z_i^{-n}\notag\\
&= \sum_{i=1}^{h+2} \frac{1-pz_i}{q\left(1+\left(1-z_i\right)(h+1)\right)}z_i^{-n-1} \notag\\ 
&=\C(n,h)+ O\left(hMp^{n+1}\right), \label{Anh}\end{align}
where $M=\max_{i=3,\dots,h+2}\left|\frac{1-pz_i}{q \left(1+\left(1-z_i\right)(h+1) \right)}\right|=O(1)$ (note that the summand involving $z_2=1/p$ cancels), %=0 for i=2
and where $\C(n,h):=\frac{1-pz_1}{q\left[1+\left(1-z_1\right)(h+1)\right]}z_1^{-n-1}$ is the contribution coming from the pole~$z_1$.
\begin{figure}[bh]
\setlength\abovecaptionskip{2mm}
\setlength\belowcaptionskip{0mm}
\includegraphics[width=.3596\textwidth]{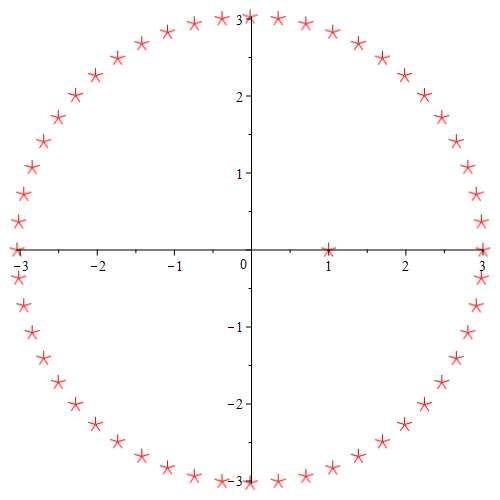}
\caption{The roots of $D(z)=1-z+q p^{h+1} z^{h+2}$ (here, with $p=1/3$ and $h=51$).
For large $h$, $D(z)$ has one dominant root $z_1$ just after~$1$, one root at $z=1/p$, and the other roots have a slightly larger modulus, 
all asymptotically close to the circle $|z|=1/p$; see Lemma~\ref{Lemma1}.}\label{roots}
\end{figure}
%Maple code : 
% restart: D:=1-z+q*p^(h+1)*z^(h+2): p:=1/3: q:=1-p: h:=51:
% sols:=[fsolve(D,z,complex)]: with(plots): complexplot(sols,style=point,symbol=asterisk,symbolsize=22,color=red);

Set $z_1:=1+\varepsilon_h$. Then $D(z_1)=1-(1+\varepsilon_h)+qp^{h+1}\left(1+\varepsilon_h\right)^{h+2}=0$, thus this implies
$\varepsilon_h=qp^{h+1}\left(1+\varepsilon_h\right)^{h+2}$; therefore we have
$z_1=1+\varepsilon_h =1+ qp^{h+1}+O(hp^{2h})$.
Now, for $h=h(n)$ tending to $+\infty$, this entails that the contribution~$\C(n,h)$ of the pole~$z_1$ (as given by Equation~\eqref{Anh}) satisfies
\begin{align}
\!\!\!\!\!\!\!\!\!\!\!\!\!\!\C(n,h)&= \frac{1-p^{h+2}+O\left(hp^{2h}\right)}{1-(h+1)qp^{h+1}+O(h^2p^{2h})}(1+\varepsilon_h)^{-n-1}\\
&=\Big(1+q(h+1)p^{h+1}-p^{h+2}+O(h^2p^{2h})\Big)\exp\left( (n+1) \ln\left(\frac{1}{1+\varepsilon_h}\right)\right)\\
&=\left(1+q(h+1)p^{h+1}-p^{h+2}+O(h^2p^{2h})\right) \exp\left( - (n+1) \varepsilon_h + \Theta((n+1)\varepsilon_h^2) \right)\!.
\label{eq:z1nh}
\end{align}
%Now, let us introduce the notation $\LN{m}(n)$ for the $m$-th iterated logarithm:
%\begin{equation*} \LN{m}(n):=\ln(\LN{m-1}(n)) \text{ with } \LN{1}(n)=\ln(n).\end{equation*}
%Observe that 
%\begin{equation} 
%\text{ if \qquad } h=c \frac{\ln(n) }{\ln(1/p) } + c' \frac{\LN{r+1}(n) }{\ln(1/p)} \text{ \qquad then \qquad } 
%p^h= \frac{ 1}{n^c \LN{r}(n)^{c'}}. \label{ph}\end{equation}
Observe that
\begin{equation} 
\text{ if \qquad } h=c \frac{\ln(n) }{\ln(1/p) } + c' \frac{\ln(\ln(n))}{\ln(1/p)} \text{ \qquad then \qquad } 
p^h=\frac{ 1}{n^c \ln(n)^{c'}}. \label{ph}\end{equation}
(Here and in the sequel we always consider $c>1/2$ and $c'\geq 0$. In fact, $c'>0$ is not needed right now, but this will be required for
the asymptotics of the mean of $H_n$ in Section~\ref{Sec4}.)

For such values of $h$, the asymptotics of the first factor in Equation~\eqref{eq:z1nh} is 
\begin{align}
1+q(h+1)p^{h+1}-p^{h+2}+O(h^2p^{2h}) 
&=1+ O\left(\frac{1}{n^c \ln(n)^{c'-1} }\right), \label{errorA}
\end{align}\def\myTheta{\Theta(n^{1-2c}/\ln(n)^{2c'})} 
and the asymptotics of the second factor in Equation~\eqref{eq:z1nh} is 
\bgroup
\setlength{\abovedisplayskip}{3pt}
\setlength{\belowdisplayskip}{-7pt}
\begin{align}
&\exp\left( - (n+1) \varepsilon_h + O((n+1)\varepsilon_h^2) \right) = \exp\left( - n qp^{h+1} + O(nh p^{2h})-\varepsilon_h +\myTheta\right) \\[-4.5mm]
&\qquad = \exp\left( - n qp^{h+1}\right) \left(1+ O(n^{1-2c } \ln(n)^{1-2c'})- O(n^{-c} \ln(n)^{-c'}) +\myTheta)\right). \label{errorB}
\end{align}
\egroup
In this expansion, one now has to check which error term dominates. 
It is the big-oh term with $n^{-c}$ if $c>1$ and the big-oh with $n^{1-2c}$ if~$c\leq 1$.
Multiplying with the asymptotic expansion from Equation~\eqref{errorA} and using the approximation~\eqref{Anh},
we get the following result (in which we simplified the $\ln$ part of the error term in a non-optimal way which will be enough for our purpose):
\begin{equation}
\Pr \left(H_n\le h\right) =\exp\left(-nqp^{h+1}\right)\left( 1+
O\left(\frac{\ln n}{n^{\min(c,2c-1)}}\right)\right). \label{Znh}
\end{equation}
Moreover, this approximation holds for all $h\in [0,n]$: first, for $h\ll \frac{1}{2} \ln(n)/\ln(1/p)$ this follows from the fact that 
 $\Pr \left(H_n\le h\right)$ is increasing with respect to $h$, and then for $h\gg c \ln(n)$
this follows from the bound~\eqref{sigma3} hereafter.

In conclusion, 
for $h=\left\lfloor\frac{\ln n}{\ln(1/p)}\right\rfloor +k$,
for any $k$ such that $h \in \left[ c_1\frac{ \ln(n)}{\ln(1/p)}, c_2\frac{ \ln(n)}{\ln(1/p)} \right]$ (with $1/2<c_1<c_2$),
we have uniformly in $k$ (when $n\to +\infty$):
\begin{align*}
\Pr \left(H_n\le h\right) &= \exp\left(-nqp^{\left\lfloor\frac{\ln n}{\ln(1/p)}\right\rfloor +k+1}\right)\left( 1+\errortermLNthree\right)
\\
&=\exp\left(-qp^{- \{\frac{\ln n}{-\ln p}\}+k+1}\right)\left( 1+\errortermLNthree\right), 
\end{align*}
and we get Theorem~\ref{Th:height_distribution} by setting $\alpha(n):=p^{-\{\frac{\ln n}{-\ln p}\}}$.
\end{proof}

\def\lg{\operatorname{lg}}
If $p=q=1/2$, we have $\alpha(n) = 2^{\{\lg(n)\}}$ (where the symbol $\lg$ stands for the binary logarithm, 
$\lg(x)=\log_2(x)$).
This subcase of particular interest corresponds to a problem initially considered in 1946 by Burks, Goldstine, and 
von Neumann~\cite{BurksGoldstinevonNeumann1946}: the study of carry propagation in computer binary arithmetic;
it constitutes one~of~the~first~analyses~of~the~cost of an algorithm! 
They gave crude bounds which were deeply improved by Knuth in 1978~\cite{Knuth1978}.
This problem can also be seen as runs in binary words, and, as such,
is analyzed by Flajolet and Sedgewick~\cite[Theorem V.1]{FlajoletSedgewick2009}.
Therein, the analysis unfortunately contains a few typos which affect some of the error terms.
Our proofs are incidentally fixing this issue. 
\smallskip

These extremal parameters (runs, longest carry) are archetypal examples of problems leading to a Gumbel distribution (or a discrete version of it). 
This distribution indeed often appears in combinatorics as the distribution of parameters encoding a maximal value: e.g.,
maximum of i.i.d.~geometric distributions~\cite{SzpankowskiRego1990},
longest repetition of a pattern in lattice paths~\cite{ProdingerWagner2007},
runs in integer compositions~\cite{Gafni2015},
carry propagation in signed digit representations~\cite{HeubergerProdinger2003},
largest part in some integer compositions, longest chain of nodes with a given arity in trees, maximum degree in some families of trees~\cite{ProdingerWagner2015},
the maximum protection number in simply generated trees~\cite{HeubergerSelkirkWagner2022}.
For some of these examples, it was proven only for some specific families of structures, but there is no doubt that it holds generically.
A general framework leading to such double exponential laws is given by Gourdon~\cite[Theorem~4]{Gourdon1998} 
for the largest component in supercritical composition schemes (see also Bender and Gao~\cite{BenderGao2014}). 
We refer to Figure~\ref{Fig:GumbelExamples} for an illustration of some of these parameters.

\begin{figure}[h]
\begin{tabular}{c|c|c}
\setlength{\tabcolsep}{-1mm}
\hspace{-5mm}
\begin{tabular}{c}
Longest up run\\  in Dyck paths\\ 
\includegraphics[width=.32\textwidth]{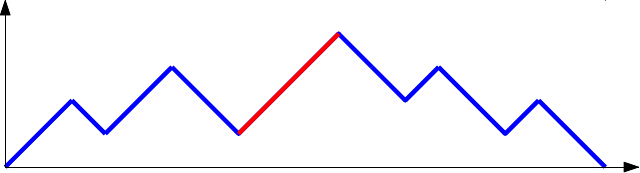} 
\end{tabular}
&
\begin{tabular}{c}
Longest chain of unary nodes\\
\includegraphics[width=.32\textwidth]{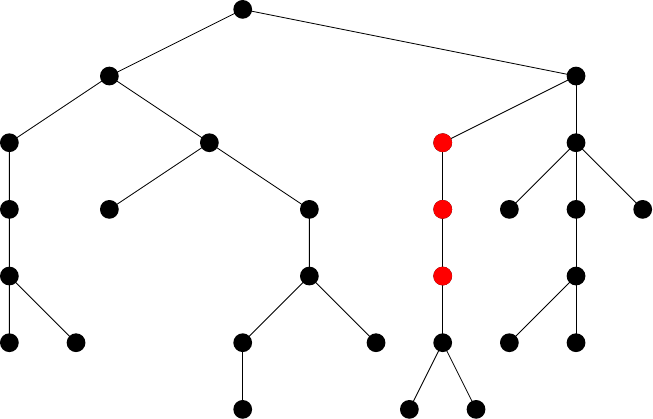} 
\end{tabular}
&
\begin{tabular}{c}
\hspace{-2.5mm}Largest part in \\
\hspace{-2.5mm}integer compositions:\\
\hspace{-2.5mm} $100= 11+1+11+9$\\
\hspace{-2.5mm}$\qquad +\textcolor{red}{39}+14+15$.
\end{tabular}
\\
\hline
\hspace{-5mm}
\begin{tabular}{c}
Longest plateau \\in Motzkin paths\\ 
\includegraphics[width=.32\textwidth]{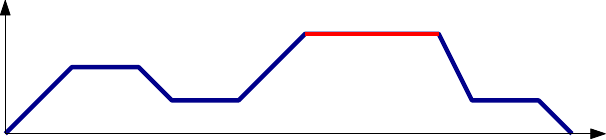} 
\end{tabular}
&
\begin{tabular}{c}
Maximal protection\\ number in trees\\
\includegraphics[width=.32\textwidth]{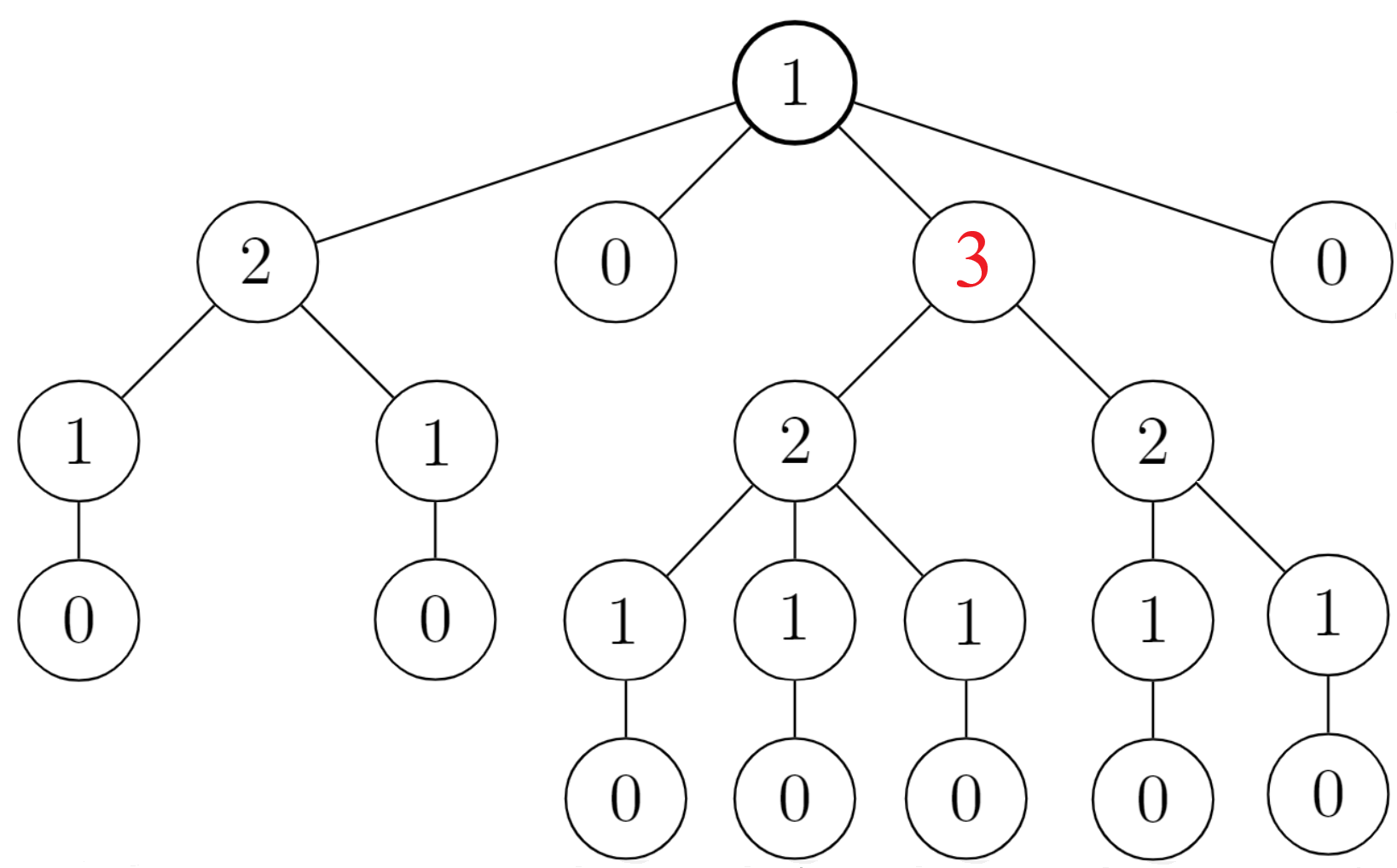} 
\end{tabular}
&
\begin{tabular}{c}
\hspace{-2.5mm}Longest run in\\
\hspace{-2.5mm}integer compositions:\\
\hspace{-2.5mm}$\!\!\!\!20= 1+4+4+1$\\
\hspace{-2.5mm}$\qquad \ \,+\textcolor{red}{3+3+3}+1$.\\
\end{tabular}
\end{tabular}
\caption{Many combinatorial structures have some parameters which asymptotically follow a discrete Gumbel distribution.}\label{Fig:GumbelExamples}
\end{figure}

The Gumbel distribution is also called the ``double exponential distribution'', or the ``type-I generalized extreme value distribution'',
and can also be expressed as a subcase of the Fisher--Tippett distribution. Let us give a formal definition.

\bgroup
\setlength{\abovedisplayskip}{4pt}
\setlength{\belowdisplayskip}{6pt}
\begin{definition}[Gumbel distribution]
A continuous random variable $X$ with support $[-\infty,+\infty]$ 
follows a \textit{Gumbel distribution}
(of parameters $\mu$ and $\beta$), denoted by $\operatorname{Gumbel}(\mu,\beta)$, if
\begin{equation*} \Pr(X\leq x ) = \exp\left(-\exp\left(-\frac{x-\mu}{\beta}\right)\right).\end{equation*}
Its mean satisfies
$\E[X] = \mu+\gamma \beta$ (where $\gamma=0.5772\dots$ is Euler's constant)
and its variance satisfies $\Var[X]= \frac{\pi^2}{6} \beta^2$.
It is unimodal with a peak at $x=\mu$ and its median is at $x=\mu-\beta \ln(\ln(2))$.
%\end{equation}
\end{definition}

\begin{definition}[Discrete Gumbel distribution]
A discrete random variable $Y$ follows a \textit{discrete Gumbel distribution} of parameters $\mu$ and $\beta$, which we denote $\operatorname{Gumbel}(\mu,\beta)$\footnote{With a slight abuse of notation, we use the same notation $\operatorname{Gumbel}(\mu,\beta)$ 
for both the continuous distribution and the discrete distribution, adding the right adjective if needed to remove any ambiguity.}, if
\begin{equation} \Pr( Y \leq h) = \exp\left(-\exp\left(-\frac{h-\mu}{\beta}\right)\right), \text{\qquad for all $h\in \Z$}.\label{discrete_gumbel}\end{equation} 
In particular, one can always write $Y=\lceil X \rceil$, where $X$ follows a continuous $\operatorname{Gumbel}(\mu,\beta)$;
note on the other side that $\lfloor X \rfloor$ follows a discrete $\operatorname{Gumbel}(\mu-1,\beta)$.
\end{definition}

To obtain a nice formula for the mean and variance of a discrete Gumbel distribution 
remains an open problem:
for example, for $Y \law \operatorname{Gumbel}(0,1)$, we have
\begin{equation*}\E[Y]=\sum_{h=-\infty}^\infty h \left(\exp(-\exp(-h))-\exp(-\exp(-h+1)\right)= 1.077240905953631072609\dots\end{equation*} 
(and it takes 5 seconds to get thousands of digits, as the terms decrease doubly exponentially fast),
but will anybody find a closed form for this mysterious constant?
Some insight on the variance of the discrete distribution $Y$ can be obtained from the continuous distribution $X$
via the following trivial but useful bounds which hold more generally as soon as  $|X-Y|<1$: 
\begin{equation}\label{EVbounds}
\left|\E[Y]-\E[X]\right| < 1 \text{\qquad and \qquad} \left|\Var[Y]-\Var[X]\right|< 2+4 |\E[X]|.\end{equation}
%\smallskip

We can now restate our previous theorem in terms of this discrete Gumbel distribution.

\begin{corollary}[Gumbel limit law]
The sequence of random variables $\lceil H_n-\frac{\ln(pq n)}{\ln (1/p)} \rceil$ 
converges for $n\rightarrow +\infty$ (in distribution and in moments) to the discrete $\operatorname{Gumbel}(0,\beta)$ distribution with $\beta=\frac{1}{\ln(1/p)}$.
Accordingly, it implies that
\begin{equation*}\E[H_n]\sim \frac{\ln(p q n)}{\ln (1/p)} + \gamma \beta+ \text{an error smaller than $1$},\end{equation*}
\begin{equation*}\Var[H_n]\sim \frac{\pi^2}{6 \ln(p)^2} + \text{ an error smaller than $2+4\gamma \beta$}.\end{equation*}
\end{corollary}
\begin{proof} 
Consider the sequence of random variables  $Y_n:=\lceil H_n-\mu_n \rceil$.
Then, the change of variable $h \mapsto h+\mu_n$ in Equation~\eqref{PHn}, with $\mu_n=\frac{\ln(pq n)}{\ln (1/p)}$ 
allows us to match $Y:=\lim_n Y_n$  (where the limit is in distribution)
with the discrete Gumbel defined in~\eqref{discrete_gumbel},
for $\mu=0$ and $\beta=\frac{1}{\ln(1/p)}$. 
Due to the exponentially small uniform error term in~\eqref{PHn} on the support $[0,n]$ of $H_n$,  we have a convergence in moments of $Y_n$ to $Y$.
Then, the asymptotics of the moments follow by applying the bounds~\eqref{EVbounds} on the link between the mean/variance of the discrete and continuous Gumbel distribution.
\end{proof}
These moment asymptotics already constitute a notable result (falling as a good ripe fruit!), 
but a very interesting phenomenon is hidden in these imprecise errors terms: some bodacious fluctuations, that we fully describe in Section~\ref{Sec4}.
\egroup

\subsection{Waiting time}
Let us end this section with an application to a natural statistic: the waiting time $\tau_h$, i.e.,~the number of steps spent by the random walk when it reaches a given altitude~$h$ for the first time.
There is an intimate relationship between height and waiting time 
(stated more formally in Equation~\eqref{wait} hereafter);
it is thus natural that they have enumerative and asymptotic formulas of a similar nature,
as better shown by the following corollary.

\begin{corollary}%[Waiting time for height~$h$]
\label{Th:stoppedtime_distribution}
The waiting time $\tau_h$ for reaching height~$h$ satisfies
\begin{align}\label{WT}
\Pr(\tau_h=n)&=[z^n] \frac{(1-pz) p^h z^h}{1 -z+ q p^{h-1} z^h }.
\end{align}
The distribution function of $\tau_h$ satisfies
\begin{equation}\label{tau}
\Pr(\tau_h\leq n) =1-\exp\left(-q\alpha(n)^2np^h\right)+\errorterm.
\end{equation}
\end{corollary}
\begin{proof}
Consider a walk reaching for the first time altitude $h$ at time $n$.
Cut it after each reset.
It gives a sequence of factors of length $k\leq h$, followed by a last factor with $h$ up steps.
This translates into the combinatorial formula
\begin{align*}
\Pr(\tau_h=n)&= [z^n] \frac{p^h z^h}{1-\sum_{k=1}^{h-1} p^{k-1} q z^k},
\end{align*}
which simplifies to Formula~\eqref{WT}.
Now, for the distribution function, 
instead of redoing a full analysis based on a partial fraction decomposition of this generating function, 
it is more convenient to use the relation
\begin{align}\label{wait} 
\Pr(\tau_h=n) = \Pr(H_n=h \text{ and } H_{n-1}<h),
\end{align} 
thus this waiting time also satisfies
\begin{align} \label{eq7}
\Pr(\tau_h\leq n) = \Pr(H_n \geq h)%=1- \Pr(H_n< h)
=1- \Pr(H_n\leq h-1).
\end{align} 
Then, using Theorem~\ref{Th:height_distribution}, we also have
\begin{align*}
\Pr(H_n\leq h-1)&=\Pr\left(H_n\leq \left\lfloor\frac{\ln n}{\ln(1/p)}\right\rfloor +h-1- \left\lfloor\frac{\ln n}{\ln(1/p)}\right\rfloor\right)
\\[-.5mm]
&=\exp\left(-q\alpha(n)p^{h- \left\lfloor\frac{\ln n}{\ln(1/p)}\right\rfloor }\right)+\errorterm
\\[-.5mm]
&=\exp\left(-q\alpha(n)^2p^{h+\frac{\ln n}{\ln p}}\right)+\errorterm.
\end{align*}
Via Formula~\eqref{eq7} linking the waiting time $\tau_h$ and the height $H_n$, this entails~\eqref{tau}.
\end{proof}
We now turn to a finer analysis of the mean and variance of $H_n$.   %Interesting to consider: consequences for culminating path.

\section{Mean and variance of the height}\label{Sec4}

\subsection{Fundamental properties of the Mellin transform}

In order to get a fine estimation of the average height, we use a Mellin transform,
which, as we shall see, is the key tool to handle the corresponding asymptotics.
We now present the needed definitions and formulas.
We refer e.g.~to~Flajolet, Gourdon, and Dumas~\cite{FlajoletGourdonDumas1995} 
or to the book \textit{Analytic Combinatorics}~\cite[Appendix B.7]{FlajoletSedgewick2009}
for more on the Mellin~transform~and numerous applications to asymptotics of harmonic sums, digital sums, and divide-and-conquer recurrences.

\begin{definition}[Mellin transform]\label{Mellin}
Let $f(t)$ be a continuous function defined on the positive real axis $0< t<+\infty$. 
The Mellin transform $f^*$ of $f$ is the function defined by 
\begin{equation*}f^*(s):=\int_{0}^{+\infty}f(t)t^{s-1}dt.\end{equation*}
This integral exists only for $s$ such that the function $f(t)t^{s-1}$ is integrable on $\left(0,\; +\infty\right)$. 
Thus, if there exist two real numbers $a$ and $b$, such that $a>b$ and 
\begin{equation}\label{strip}
f(t)=\begin{cases}
O(t^{a} ), & \mbox{ if } t \to 0
\\
O(t^b), & \mbox{ if } t \to +\infty 
\end{cases}, 
\end{equation}
then the function $f^*$ is well defined for any complex number $s$ with real part such that $-a<\Re(s)<-b$; this domain is called the fundamental strip of $f^*$.
Moreover, for all~$c$ in this domain, if $f^*(s)$ converges uniformly to 0 for $s=c\pm i\infty$, then
the function $f$ can be expressed for $t\in(0,+\infty)$ as the following inverse Mellin transform:
\begin{equation}f(t)=\frac{1}{2i\pi}\int_{c-i\infty}^{c+i\infty}f^*(s)t^{-s}ds. \label{invMellin}\end{equation}
\end{definition}

As an example, let us consider the gamma function, which illustrates well the role of the fundamental strip (and this example will also play a role in the next pages).
\begin{example}[The gamma function as a Mellin transform]\label{example1}
The gamma function satisfies 
\begin{align}
\Gamma(s)&=\int_0^{+\infty} \exp(-t) t^{s-1} dt \text{\qquad (for $0<\Re(s)<+\infty$)}, \label{Gamma1}\\
\Gamma(s)&=\int_0^{+\infty} \left(1-\exp(-t)\right) t^{s-1} dt \text{\qquad (for $-1<\Re(s)<0$)}. \label{Gamma2}
\end{align}
\end{example}

An important consequence of Formula~\eqref{invMellin} is that,
if $f$ is a meromorphic function on~$\mathbb{C}$,
and if $\lim_{c\rightarrow +\infty} \int_{c-i\infty}^{c+i\infty}f^*(s)t^{-s}ds=0$,
then one can push the integration contour of For\-mula~\eqref{invMellin} 
to the right (taking $\lim_{c\rightarrow +\infty}$) 
and one then collects in passing the contributions %coming 
from the residue at each pole $s_k$ to the right of the fundamental strip. 
Now, for $t>0$ and $a\!\in\!{\mathbb C}$,
multiplying  $t^{-s}=t^{-a} \sum_{\ell\geq 0}\ln(t)^\ell (a-s)^\ell /\ell!$ by the Laurent series of $f^*(s)$ at $s\!=\!s_k$, we see that
$\Res[ f^*(s) t^{-s},s_k]$ can be expressed\footnote{The notation $\Res[g(s),s_k]$ stands for the residue of $g(s)$ at $s=s_k$.} as a sum of $\operatorname{order}(s_k)$ terms, and~one~gets%\linebreak
\begin{align}
f(t)= &\sum_{\substack{\text{$s_k$ pole of $f^*(s)t^{-s}$}\\ \text{$\Re(s_k)\geq -b$}}} \Res[f^*(s)t^{-s},s_k]\\
= &\sum_{\substack{\text{$s_k$ pole of $f^*$}\\ \text{$\Re(s_k)\geq -b$}}} 
\sum_{j=1}^{\operatorname{order}(s_k)}
\Res[(s-s_k)^{j-1} f^*(s),s_k] \ t^{-s_k} \frac{(-1)^j}{(j-1)!} \, (\ln t)^{j-1}.\label{invMellin2}\end{align}

\subsection{Average height of Moran walks}

We now state the main result of this section.

\begin{theorem}[Average height]\label{Thmean}
The average height of Moran walks of length~$n$ is~given~by
\begin{align}\label{Ehn}
\E[H_n]&=\frac{\ln n}{\ln(1/p)}-\frac{\gamma}{\ln p}-\frac{1}{2}-\frac{\ln q}{\ln p}+\frac{Q(\ln(q n))}{\ln p}+
\errortermLNfour,
\end{align}
where $\gamma = .57721\dots$ is Euler's constant, and 
where $Q$ is an oscillating function (a Fourier series of period $\ln(1/p)$) given by\vspace{-.6mm}
\begin{align}\label{defQ}
Q(x)&:=\sum_{k\in\mathbb{Z}\setminus\{0\}}\Gamma(s_k) \exp(-s_k x)
\text{\quad where $s_k:=\frac{2ik\pi} {\ln p}$}. 
\end{align}
\end{theorem}\vspace{-3.2mm}
\begin{remark}[Fourier series representation]\label{RemarkReal}
The fact that $Q$ is a Fourier series of period $\ln(1/p)$ and is real for $x\in\R$ is better seen 
via the alternative equivalent expression
\begin{align*}
 Q(x) &= 2 \sum_{k\geq 1} \left( \Re(\Gamma(s_k)) \cos\left(\frac{2 k \pi x}{\ln(p)}\right)+ \Im(\Gamma(s_k)) \sin\left(\frac{2 k \pi x}{\ln(p)}\right) \right),
\end{align*}
where $\Re$ and $\Im$ stands for the real and imaginary parts. This is illustrated in Figure~\ref{fig:Q}.
\end{remark}\vspace{-3.5mm}
\begin{remark}[Fourier series differentiability]\label{rem:Fourier}
Such asymptotics involving fluctuations dictated by a Fourier series are typical of results obtained via Mellin transforms. 
They often appear in the asymptotic cost of divide-and-conquer algorithms, 
or of expressions involving digital sums, harmonic sums, or finite differences (see the work of de Bruijn, Knuth, and Rice~\cite{deBruijnKnuthRice1972,Knuth1978}, or Flajolet, Gourdon, and Dumas~\cite{FlajoletGourdonDumas1995}).
It is sometimes also possible to get them via some real analysis (like Pippenger did~\cite{Pippenger2002}), or like in the seminal work of Delange~\cite{Delange1975} on the sum of digits. 
Note that the Delange series is nowhere differentiable, while our Fourier series is infinitely differentiable, as proven in Theorem~\ref{Cinfty}.
\end{remark}\vspace{-4mm}

\begin{figure}[b]
\begin{tabular}{ccc}
\begin{tabular}{c}
\includegraphics[width=.3601\textwidth]{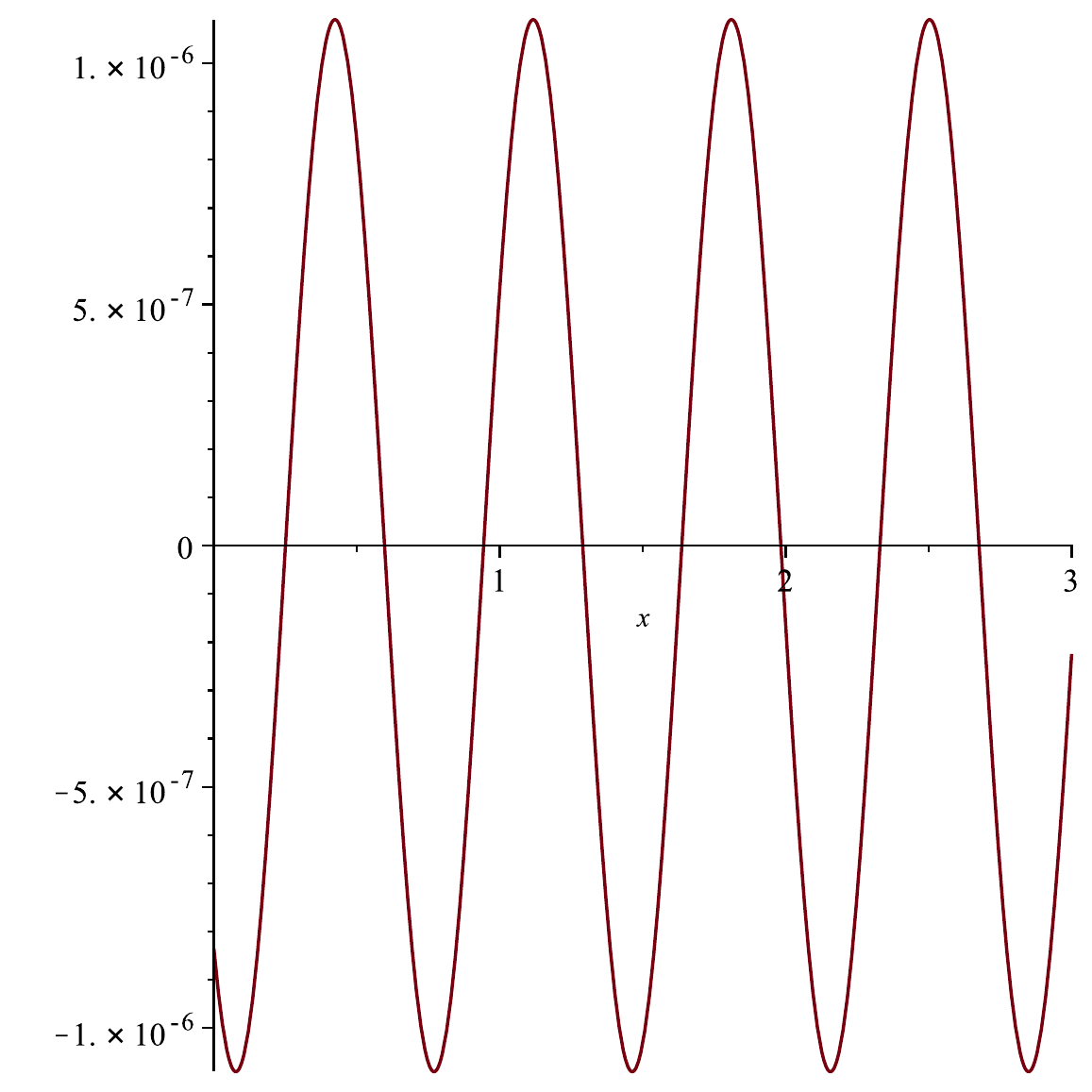}\\
$Q(x)$ (for $p=1/2$)
\end{tabular}
&
\begin{tabular}{c}
\qquad
\end{tabular}
&\begin{tabular}{c}
\includegraphics[width=.3601\textwidth]{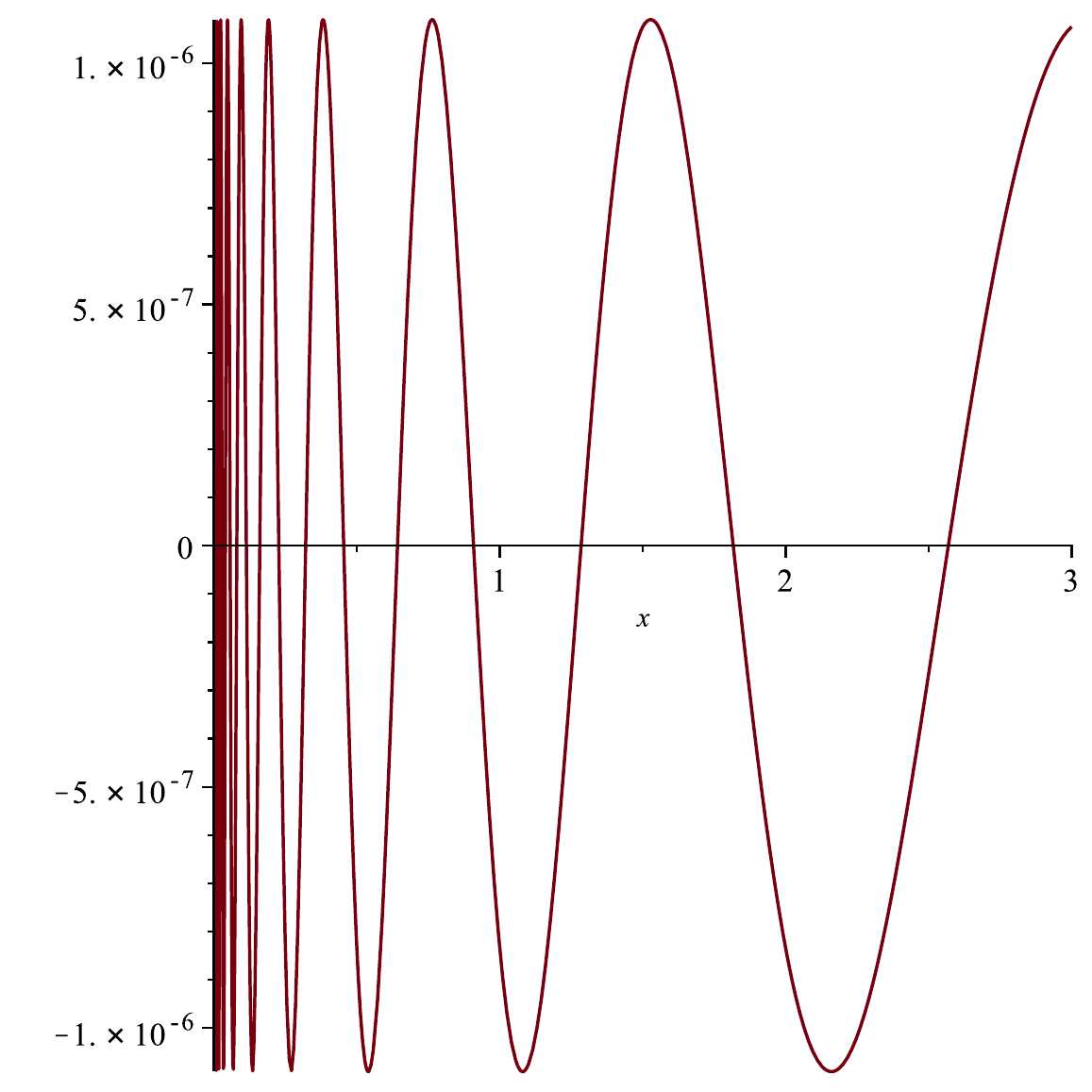} \\
 \qquad $Q(\ln(px))$ (for $p=1/2$)
\end{tabular}
\end{tabular}
\setlength\abovecaptionskip{2mm}
\caption{The height of Moran walks involves asymptotic fluctuations encoded by a Fourier series $Q(x)$, of period $\ln(1/ p)$, and weak amplitude.
More precisely, it involves $Q(\ln(px))$ which thus oscillates an infinite number of times for $x \rightarrow 0^+$, 
and these oscillations get larger and larger for $x \rightarrow +\infty$.
Moreover, $Q$ oscillates faster when $p$ tends to~$1$. We shall encounter later another Fourier series,
$R(x)$, which shares all these properties.
}\label{fig:Q}
\end{figure}

\begin{proof}[Proof of Theorem~\ref{Thmean}]
The proof exploits the fact that the mean $\E[H_n]$
asymptotically behaves like $\sum_{h=0}^{+\infty} \big( 1- \exp(-nqp^{h+1})\big)$;
this is proven by rewriting $\E[H_n]$ as follows:\vspace{-0.9mm}
\bgroup
\setlength{\abovedisplayskip}{1.95pt}
\setlength{\belowdisplayskip}{1.95pt}
\begin{equation}\E[H_n]=\sum_{h= 0}^n \left(1- \Pr\left(H_n\leq h\right) \right) = \Sigma_0 + \Sigma_1 +\Sigma_2 +\Sigma_3 - \Sigma_4 + \Sigma_\infty,
\label{partition}\end{equation}
with\vspace{-1.55mm}
\begin{align*}
\Sigma_0&:= \sum_{0 \leq h < h_1} \left(\exp(-nqp^{h+1})-\Pr\left(H_n \leq h\right)\right), \\[-0.5pt]
\Sigma_1&:= \sum_{h_1 \leq h < h_2} \left(\exp(-nqp^{h+1})-\Pr\left(H_n \leq h\right)\right),\\[-0.5pt]
\Sigma_2&:= \sum_{h_2 \leq h < h_3} \left(\exp(-nqp^{h+1})-\Pr\left(H_n \leq h\right)\right),\\[-0.5pt]
\Sigma_3&:= \sum_{h_3 \leq h \leq n} \big(1-\Pr\left(H_n \leq h\right)\big), \\[-0.5pt]
\Sigma_4&:= \sum_{h= h_3}^{+\infty} \big( 1- \exp(-nqp^{h+1})\big), \\[-0.5pt]
\Sigma_\infty&:= \sum_{h=0}^{+\infty} \big( 1- \exp(-nqp^{h+1})\big).
%\vspace{-3mm}
\end{align*}

The key is to prove that, for some $h_1$, $h_2$, and $ h_3$ adequately chosen, the sums $\Sigma_0, \Sigma_1, \Sigma_2$, $\Sigma_3$, and~$\Sigma_4$ are asymptotically negligible,
while the main contribution to $\E[H_n]$ comes from the last sum (namely, $\Sigma_\infty$),
which we will evaluate via a Mellin transform approach.

%where $h_1=c_1\ln(n)/\ln(1/p)$ and $h_2=c_2\ln(n)/\ln(1/p)$ with $c_1=3/4$ and $c_2=2$.
%Let us introduce $\Lambda(n)$, a function tending arbitrarily slowly to $+\infty$, for example 
%the reader can have in mind $\Lambda(n):=\ln(\ln(\ln(\ln(n))))$.

The reader not enjoying delta-epsilon proofs
could have the feeling that ``cutting epsilons into 5 parts''
 like above is a little bit discouraging but this is the price to pay to get the $O((\ln(n)^4/n)$ error term in Formula~\eqref{Ehn}.
In fact, in Equation~\eqref{partition} for $\E[H_n]$, it is possible to cut the sum into only 4 parts, but then 
this would lead to a final weaker $O(1/\sqrt n)$ error term.

So let's be brave and begin with $\Sigma_0$. Here, for the range $0 \leq h<h_1$, with $h_1:=\frac{3}{4}\frac{\ln(n)}{\ln(1/p)}$,\newline
 we get\vspace{-3.1mm}
\begin{align*}
| \Sigma_0 | &\leq h_1 \times \left( \max_{0 \leq h < h_1} \left(\exp(-nqp^{h+1}) + \max_{0 \leq h < h_1} \Pr\left(H_n \leq h\right)\right) \right)\\
& =h_1 \times \left( \exp(-nqp^{h_1+1}) + \Pr\left(H_n \leq h_1\right)\right) \\
& =h_1 \times \left( 2 \exp(-qp n^{1/4})+ \errortermLNthree \right) \\
& = \errortermLNfour, \end{align*}
where, for the second line we used that the sequences are increasing with respect to~$h$, 
and for the third line we used Formula~\eqref{ph} for $p^h$ and the approximation of Theorem~\ref{Th:height_distribution}.
Note that this bound for $| \Sigma_0 |$ also implies the uniform bound
\begin{equation}
\Pr(H_n\leq h)= \errortermLNfour \qquad \text{(for $h<h_1$)}. \label{PHnsmall}
\end{equation}
\egroup

Now, for $\Sigma_1$, in the range $h_1 \leq h<h_2$, with $h_2:=\frac{\ln(n)}{\ln(1/p)}+\frac{\ln(\ln(n))}{\ln(1/p)}$,
we rewrite $h$ as $h:=(1-t) h_1 + t h_2$. Such values of $h$ correspond to using $c=(t+3)/4$ and $c'=t$
in the Formula~\eqref{ph} for $p^h$.

Via the exponential bound on $H_n$ from Formula~\eqref{Znh}, we get 
\begin{align*} | \Sigma_1 | &\leq (h_2-h_1) \times \left( \max_{h_1 \leq h < h_2} \left(\exp(-nqp^{h+1}) + \max_{h_1 \leq h < h_2} \Pr\left(H_n \leq h\right)\right) \right)\\
& \leq h_2 \times \left( \exp(-nqp^{h_2+1}) + \Pr\left(H_n \leq h_2\right)\right) = O((\ln n)^4/n). \end{align*}

Then, for $\Sigma_2$, in the range $h_2 \leq h_3$, with $h_3:= \frac{4 \ln(n)}{\ln(1/p)}$,
we rewrite $h$ as $h:=(1-t) h_2 + t h_3$. Such values of $h$ correspond to using $c=1+3t$ and $c'=1-t$
in the Formula~\eqref{ph} for $p^h$.
Via Formula~\eqref{Znh}, we get $|\Sigma_2|= O((\ln n)^3/n)$.

For the next sum, using the power series expansion of the exponential in Equation~\eqref{eq:z1nh} (and keeping in mind that our choice of $h_3$ implies $p^{h_3}= 1/n^4$), we get
\bgroup
\setlength{\abovedisplayskip}{2.3pt}
\setlength{\belowdisplayskip}{5pt}
\begin{align} \Sigma_3=\sum_{h=h_3}^n \left(1-\Pr\left(H_n\leq h\right) \right)
&\leq (n+1-h_3) \left(1-\Pr\left(H_n\leq h_3\right) \right) \\[-7pt]
&\leq n (1-\exp(-(n+1) q p^{h_3+1})) (1+o(1))
=O\left( \frac{1}{n^2}\right).\qquad \label{sigma3}\end{align}

Finally, for the sum $\Sigma_4$, 
we use the power series expansions of $\exp(x)$ and of $1/(1-p)$ and we get:
\begin{equation*}
\Sigma_4=\sum_{h\geq h_3} (1- \exp(-nqp^{h+1}))= \frac {n q p^{h_3+1}}{1-p}
- \sum_{h\geq h_3} \sum_{k\geq 2 } \frac{ (-nqp^{h+1})^k}{k!} 
< n p^{h_3 +1}=O\left(\frac{1}{n^{3} }\right).
\end{equation*}
\egroup
We got that $\Sigma_0$, $\Sigma_1$, $\Sigma_2$, $\Sigma_3$, and $\Sigma_4$ are $o(1)$.
It remains to evaluate $\Sigma_\infty=\sum_{h\ge 0}(1-e^{-nqp^{h+1}})$.
Such a sum is typical of expressions which can be evaluated by Mellin transform techniques. To this aim,
let $\phim(t)=\sum_{h\ge 0}(1-e^{-tqp^{h+1}})$ and 
set $f(t):=1-e^{-tpq}$ and $\mu_h:=p^{h}$, then \begin{equation*} \phim(t)=\sum_{h\ge 0}f(\mu_ht).\end{equation*} Let $\phim^*$ and $f^*$ be, respectively, the Mellin transform of the functions $\phim$ and $f$.
Using Identity~\eqref{Gamma2} given in Example~\ref{example1}, we have $f^*(s)=-(pq)^{-s}\Gamma(s)$ on its fundamental strip $-1<\Re(s)<0$
and, as $\phim$ is a harmonic sum, its Mellin transform is 
\begin{equation}\label{harmonicsum}
\phim^*(s)=f^*(s)\sum_{h\ge 0}\mu_h^{-s}=\frac{q^{-s} \Gamma(s) }{1-p^{s}}.
\end{equation}

This function extends analytically to the full complex plane, with isolated poles
at the negative integers (due to poles of $\Gamma(s)$ there), and with another set of isolated poles (the roots of $p^{s}=1$). These two sets of poles have $s=0$ in common.
This implies that for $\Re(s)>-1$ the poles of $\phim^*$ are 
\begin{align}\label{poles}
\begin{cases}
s_k=\frac{2ik\pi}{\ln p} \text{ for $k\in\mathbb{Z}, k\neq 0$ \qquad (all are poles of order 1)},\\
s_0=0 \qquad \text{(the only pole of order 2)}.
\end{cases}
\end{align}
Using Formula~\eqref{invMellin2} for the inverse Mellin transform, we obtain
\begin{align*}
\phim(t)&=
	\Res[s\phi^*,0] \ \ln t - \Res[\phi^*,0] - \sum_{k\in \Z\setminus\{0\}} \Res[\phi^*,s_k] \ t^{-s_k} \\[-0.5pt]
&=\frac{\ln t}{-\ln p} -\left(\frac{\gamma}{\ln p}+\frac{1}{2}+\frac{\ln q}{\ln p}\right) +
\frac{1}{\ln p}\sum_{k\in\mathbb{Z}\setminus\{0\}}\Gamma(s_k)q^{-s_k}t^{-s_k}. 
\end{align*}
%#-Residue phi* en 0:
%series(GAMMA(s)*q^(-s)/(1-p^(s)),s=0,3);
%-coeff(%,s,-1): expand(%);
We finally get the claim of the theorem by noting that $\E[H_n]=\phim(n)+O\left(\frac{(\ln n)^4}{n}\right)$.
\end{proof}

\subsection{Variance of the height of Moran walks}
We now prove that the height of Moran walks, 
despite a mean of order $O(\ln n)$ and a second moment of order $O((\ln n)^2)$,
has a variance which involves surprising cancellations at these two orders, 
leading to an oscillating function of order $O(1)$ (in $n$), 
as implied by the following much more precise asymptotics.
\begin{theorem}\label{Thvar}
The variance of the height of Moran walks satisfies
\begin{equation*}
\Var[H_n]=\frac{1}{\ln(p)^2} \left(Q^2(\ln(qn))+2\gamma Q(\ln(qn))+ 2R(\ln(qn)) + \frac{\pi^2}{6} \right) + \frac{1}{12} +O\left(\frac{(\ln n)^5}{ n}\right),
\end{equation*}
where $Q$ and $R$ are Fourier series of small amplitudes given by Formulas~\eqref{defQ} and~\eqref{defD}. 
\end{theorem}
\begin{proof}
To obtain the variance of $H_n$ we first consider 
the second moment 
\begin{align}\label{Ehn2}
\E[H_n^2]&= \sum_{h\geq 0} \Pr(H_n=h) h^2 = \sum_{h\geq 0} \Pr(H_n^2> h ),
\end{align} 
where we know from Theorem~\ref{Th:height_distribution} that the summand can be approximated by
\begin{equation*}
\Pr(H_n^2> h )=%1-\Pr\left(H_n^2\le h\right)=
1-\Pr\left(H_n\le \sqrt h\right)=1-\exp\left(-nqp^{\left\lfloor\sqrt h\right\rfloor+1}\right) +\errorterm.
\end{equation*} 
Then, partitioning the last sum in~\eqref{Ehn2} into 
the same intervals as in Formula~\eqref{partition},
we get that 
$\E[H_n^2]= \phiv(n) +O\left(\frac{(\ln n)^4}{n}\right)$, where 
$\phiv$ is the function defined by
\begin{equation*}
\phiv(x)=\sum_{h\ge 0}\left(1-\exp\left(-xqp^{\left\lfloor\sqrt h\right\rfloor +1}\right)\right).
\end{equation*}
From the behavior of $\phiv(x)$ at $x=0$ and $x=+\infty$,
using the property given in~\eqref{strip}, we get that the Mellin transform of $\phiv$ is defined on the fundamental strip $(-1, \, 0)$.
Using the harmonic sum summation~\eqref{harmonicsum},
one gets for $s$ in this strip:
\begin{equation*}
\phiv^*(s)=f^*(s)\sum_{h\geq 0}\left( p^{\left\lfloor\sqrt h\right\rfloor }\right)^{-s}=-\Gamma(s)(pq)^{-s}\sum_{h\geq 0}\left( p^{\left\lfloor\sqrt h\right\rfloor }\right)^{-s}.
\end{equation*}
Here, as we have
\begin{align*}
\sum_{h\geq 0}\left( p^{\left\lfloor\sqrt h\right\rfloor }\right)^{-s}= \sum_{n\geq 0}\,\sum_{h=n^2 }^{(n+1)^2-1}\left( p^{-s}\right)^n
= \sum_{n\geq 0}\,\left(2n+1\right)\left( p^{-s}\right)^{n } =\frac{1+p^{-s}}{\left(1-p^{-s}\right)^2},
\end{align*}
we finally get
\begin{equation}\label{phivar}
\phiv^*(s)=\frac{-\Gamma(s)q^{-s}(1+p^{s})}{(p^{s}-1)^2}.
\end{equation}
What are the poles of $\phiv^*(s)$? These are $s=0$ (a pole of order 3) and $s=s_k=2ik\pi$ (for $k\in \Z, k\neq 0$, which are poles of order 2).
Using Formula~\eqref{invMellin2} for the inverse Mellin transform, one thus obtains
\begin{align}
\phiv(t)=& \frac{\ln(t)^2}{\ln(p)^2} + \ln(t) \frac{\ln(p)+2\ln(q)+2\gamma-2 Q(\ln(qt)) }{\ln(p)^2} \notag \\
&- \frac{\ln(p)+2\ln(q)}{\ln(p)^2} Q(\ln(qt)) + \frac{2}{\ln(p)^2} R(\ln(qt)) \notag \\
&+{\frac{1}{3}}+{\frac {\gamma+\ln(q)}{\ln(p)}}
+\frac {\pi^2/6+\gamma^2}{\ln(p)^2}
+\frac {2\gamma\,\ln(q)+\ln(q)^2}{ \ln (p)^2},
\end{align}
with the same $Q(x)$ as in~\eqref{defQ}, and where $R(x)$ is another Fourier series given by
\begin{align}
	R(x)&=\sum_{k \in \Z\setminus \{0\}} \Gamma'(s_k) \exp(-s_k x). \label{defD}
\end{align}
(Similarly to $Q(x)$, this Fourier series $R(x)$ is always real, as can be seen by 
replacing $\Gamma$ by $\Gamma'$ in Remark~\ref{RemarkReal}.) 

Now that we obtained the asymptotic behavior of $\E[H_n^2]$, 
we conclude and obtain Theorem~\ref{Thvar} via $\Var[H_n]=\E[H_n^2]-\E[H_n]^2$,
where $\E[H_n]$ was computed in Theorem~\ref{Thmean}.
\end{proof}

\subsection{Height of excursions}
Excursions are walks in $\N^2$ 
ending at altitude $0$ (where, as previously, time is encoded by the $x$-axis, and altitude by the $y$-axis). 
As in previous sections, let $Y_n$ and $H_n$ be the final altitude and %the 
height of a walk,
and let the random variable $\H_n$ be the height of a walk of length~$n$ conditioned to be an excursion,
 that is, $\H_n= H_n | \{Y_n=0\}$. 
For Moran walks, we get the following behavior.

\begin{theorem}[Distribution and moments of the height of Moran excursions]\label{Th:height_distribution_excursions}
The distribution of the height of excursions satisfies (for a uniform error term in $k$)
\begin{align}
\Pr\left( \H_n\le \left\lfloor\frac{\ln n}{\ln(1/p)}\right\rfloor+k\right)= \exp\left(-q\alpha(n-1)p^{k+1}\right)+\errorterm,
\end{align}
\noindent with $\alpha(n):=p^{-\{\frac{\ln n}{\ln(1/p)}\}}$ (where $\{x\}$ stands for the fractional part of $x$,
and where $\lfloor x \rfloor$ stands for the floor function of $x$).

Introducing temporarily the quantity $\ell_n:=\ln(q(n-1))$, and with the same Fourier series $Q$ and $R$ as in Theorems ~\ref{Thmean} and~\ref{Thvar},
the average and the variance are given by
\begin{align} \label{eqEHn}
\E[\H_n]&=\frac{\ln n}{\ln(1/p)}-\frac{\gamma}{\ln p}-\frac{1}{2}-\frac{\ln q}{\ln p}+\frac{Q(\ell_n)}{\ln p}+\errortermLNfour,
\end{align} 
\begin{equation}\label{eqVHn}
\Var[\H_n]=\frac{1}{\ln(p)^2} \left(Q^2(\ell_n)+2\gamma Q(\ell_n)+ 2R(\ell_n) + \frac{\pi^2}{6} \right) + \frac{1}{12} +O\left(\frac{(\ln n)^5}{n}\right).
\end{equation}
%where $Q$ and $R$ are Fourier series of small amplitudes given by Formulas~\eqref{defQ} and~\eqref{defD}. 
\end{theorem}
\begin{proof}
As a Moran excursion necessarily ends by a reset, we have 
\begin{equation}\label{exc_trivial}
\Pr(\widetilde{H}_n\leq h)= \Pr\left(H_n\leq h | \{Y_n=0\}\right) = q \Pr(H_{n-1}\leq h) /\Pr(Y_n=0).\end{equation}
Thus, we have $\Pr(\widetilde{H}_n\leq h)= \Pr(H_{n-1}\leq h)$, $\E[\H_n ]=\E[H_{n-1}]$, and $\Var[\H_n]=\Var[H_{n-1}]$,
we can therefore directly recycle the results of Theorems~\ref{Th:height_distribution},~\ref{Thmean}, and~\ref{Thvar}
to get the asymptotic distribution/mean/variance.

In this recycling, some care has to be brought while performing the substitution $n \rightarrow n-1$
in the asymptotic formulas for the walks: 
indeed, this could impact intermediate asymptotic terms (smaller than the main asymptotic term, but larger than the error term);
however, in our case, all is safe as we have 
\begin{equation*} \frac{(\ln(n \pm 1))^m}{(n\pm 1)^{m'}} = \frac{(\ln n)^m}{n^{m'}} + O\left(\frac{(\ln n)^m}{n^{m'+1}}\right).\qedhere
\end{equation*}
\end{proof}
This result is a simple consequence of the combinatorially obvious identity~\eqref{exc_trivial},
so this direct link between the asymptotics of walks and excursions 
holds in wider generality for any model of walks with resets for which the step set $\mathcal S$ contains only positive steps.
%End of \subsection{Height of excursions}

\subsection{Fourier series: bounds and infinite differentiability}

In his seminal work~\cite{Knuth1978}, Knuth mentions at the end of his Section 3 that if one assumes that $\ln(qn)$ is equidistributed mod~1, then the sum $Q(\ln(qn))$ is of ``average~0''.
Let us amend a little bit Knuth's assertion. Indeed,
Weyl's criterion asserts that a sequence~$a_n$ is equidistributed mod~1 if and only if, for any positive integer~$\ell$, we have 
\begin{equation*}
\lim_{N\rightarrow +\infty} \frac{1}{N}\sum_{n=1}^N \exp(2i\pi \ell a_n) =0.
\end{equation*}
Considering this sum with $\ell=1$ and $a_n=\ln(qn)$, 
and applying the Euler--Maclaurin formula to it, one gets that it does \textit{not} converge to 0, and therefore $\ln(qn)$ is not equidistributed mod~1.

However, it is indeed true that the oscillating $Q(x)$ and $R(x)$ are of mean value zero over their period
(i.e.,~$\int_0^{\ln(1/p)}Q(x) dx=0$; see Figure~\ref{fig:Q} on page~\pageref{fig:Q}), 
and that $Q(\ln(qn))$ and $R(\ln(qn))$ are ``almost'' of mean value zero
and that they possess small fluctuations. Let us give an explicit bound on their amplitude.
To this aim, we first need to bound the digamma function\footnote{This is a rather misleading name: indeed, the digamma function is traditionally denoted by the letter psi
(i.e.,~$\psi$),
while it should logically be denoted by the Greek letter digamma (i.e.,~$\digamma$, a letter which looks 
like a big $\Gamma$ stack on a small $\Gamma$, which later gave birth to the more familiar letter~$F$ in the Latin alphabet). 
This paradox is due to the fact that Stirling, who introduced this function,
did initially use the notation digamma~$\digamma$, but later authors switched the notation to $\psi$, while the initial name remained.}, defined by 
\begin{equation} \psi(z):=\Gamma'(z)/\Gamma(z).\end{equation} 
The function $\psi$ can be seen as an analytic continuation of harmonic numbers and satisfies $\psi(t+1)=\psi(t)+1/t$.
While several bounds for $\psi(z)$ exist in the literature (see e.g.~\cite{Yang2015}), %citetoadd Batir2011,Guo2015
 most of them are dedicated to $z\in \R$ (for example we have $\psi(t)<\ln(t)-1/(2t)$ for $t>0$),
so we now establish a lemma for $z \in i \R$ (which we believe to be new, and which has its own interest beyond our application hereafter to bounds of Fourier series).
\begin{lemma}[A bound for the digamma function on the imaginary axis]\label{LemmaPsi}
For $t>0$, we have
\begin{equation}\label{bound1}
\left|\psi(it)\right| \leq \frac{1}{2}\ln\left(1+t^2\right) + \left(\frac{\pi}{2}+ 1-\gamma\right)+\frac{1}{t},
\end{equation}
which also implies the bound
\begin{equation*}
\left|\psi(it)\right| \leq\left(\frac{\pi}{2}+ 1-\gamma+\frac{\ln 2}{2}\right)+\left(\ln (t) \indicator_{\{t\geq 1\}}+\frac{1}{t}\right).
\end{equation*}
\end{lemma}
\begin{proof}
Using Euler's representation of the gamma function as an infinite product, i.e.,
\begin{equation*} \Gamma(z)=\frac{1}{z} \prod_{k\geq 1} (1+1/k)^z/ (1+z/k) = \frac{\exp(-\gamma z)}{z} \prod_{k\geq 1} \frac{\exp(z/k)}{1+z/k},\end{equation*}
 we get that its logarithmic derivative, $\psi(z)=\Gamma'(z)/\Gamma(z)$, satisfies, for $z\in\mathbb{C}, z\notin -{\mathbb N}$ :
\begin{equation}\psi(z) =-\frac{1}{z}-\gamma+\sum_{k=1}^{+\infty}\frac{z}{k(k+z)}.\label{Psi} \end{equation}

We refer to~\cite[Section 1.1]{Erdelyi1953} for more details on these formulas.
Now, setting $z=it$ (with $t>0$), and regrouping the imaginary and real parts gives
\begin{equation*}
\psi(it)=i\left(\frac{1}{t}+\sum_{n=1}^{+\infty}\frac{t}{n^2+t^2}\right)+\left(\sum_{n=1}^{+\infty}\frac{t^2}{n\left(n^2+t^2\right)}-\gamma\right),
\end{equation*}
and thus, by the triangle inequality
\begin{equation}
\left|\psi(it)\right| \leq \left(\frac{1}{t}+\sum_{n=1}^{+\infty}\frac{t}{n^2+t^2}\right)+\left(\sum_{n=1}^{+\infty}\frac{t^2}{n\left(n^2+t^2\right)}-\gamma\right). \label{Psi2}
\end{equation}
Here, note that for all $n\leq u< n+1$, we have $ n^2+t^2\leq u^2+t^2< (n+1)^2+t^2$, and thus 
\begin{equation*}
\frac{t}{(n+1)^2+t^2}\leq \int_n^{n+1}\frac{t}{u^2+t^2}du \leq \frac{t}{n^2+t^2}.
\end{equation*}
Summing for $n$ from $0$ to $+\infty$, we obtain 
\begin{align*}
\sum_{n=1}^{+\infty}\frac{t}{n^2+t^2}\leq \sum_{n=0}^{+\infty}\int_n^{n+1}\frac{t}{u^2+t^2}du =\int_0^{+\infty}\frac{t}{u^2+t^2}du=\frac{\pi}{2}.
\end{align*}
So the first infinite sum in~\eqref{Psi2} is bounded by $\pi/2$.
For the second infinite sum, it is convenient to split it in the contribution from the summand for $n=1$, which is bounded~by
\begin{equation*}
\max_{t\geq 0}\left( \frac{t^2}{1+t^2} \right)=1,
\end{equation*}
plus the remaining part (i.e.,~the sum of the terms for $n\geq 2$):
\begin{align*}
\sum_{n=2}^{+\infty}\frac{t^2}{n\left(n^2+t^2\right)} \leq \int_{t^{-1}}^{+\infty}\frac{1}{u(u^2+1)}du \ = \ \frac{1}{2} \ln(1+t^2) .
\end{align*}
Plugging these two bounds in~\eqref{Psi2} proves our lemma.
\end{proof}

Equipped with the previous lemma, 
we can now give our bounds for $Q(x)$ and $R(x)$.

\begin{proposition}[Uniform bounds for the oscillations]
The oscillating functions $Q(x)$ and $R(x)$ are uniformly bounded by
\def\lnexp{\operatorname{lnexp}}
\begin{align}
\sup_{x\in\R^+} |Q(x)|& \leq \frac{\ln(p)}{\pi}\lnexp\left(p,\frac{4}{5} \pi^2\right),
\\
\sup_{x\in\R^+} |R(x)|& \leq \frac{\ln(p)}{\pi}\left[\lnexp\left(p,\frac{4}{5}\pi^2\right)+\left(\frac{\pi}{2}\!+\!1\!-\!\gamma\!-\!\frac{\ln (p)}{2\pi}\right)\lnexp\left(p,\frac{114}{155}\pi^2\right)\right], \qquad \label{boundR}
\end{align}
where %$\displaystyle{\lnexp(p,\beta):= \ln\left(1-\exp\left(\frac{\beta}{\ln(p)}\right)\right)}.$
\begin{equation}\lnexp(p,\beta):=\ln\left(1-\exp\left(\frac{\beta}{\ln(p)}\right)\right).\end{equation}
For $p=1/2$, we have more precisely
\begin{equation*}
\sup_{x\in\R^+} |Q(x)|= 1.090430\dots \times 10^{-6} \text{\qquad and \qquad}
\sup_{x\in\R^+} |R(x)|= 2.987768\dots \times 10^{-6}.
\end{equation*}
\end{proposition}

\begin{proof}
Applying the triangle inequality on the definition of $Q(x)$ in~\eqref{defQ}, we get
\begin{equation*}|Q(x)| \leq \sum_{k\in\mathbb{Z}\setminus\{0\}} |\Gamma(s_k)| \times |\exp(-s_k x)|
\leq 2 \sum_{k\geq 1}\left|\Gamma(s_k)\right|\end{equation*} (a quantity independent of $x$, as $|\exp(-s_kx)|=1$).
Then, using the complement formula for the gamma function, we have
$\Gamma(-z) \Gamma(z) = \frac{\pi}{z \sin(\pi (z+1))}$ (for $z\not \in \Z$).
Using this relation for $z=it$ (with $t \in \R$) together with the relation
 $\overline{\Gamma(z)} = \Gamma(\bar z)$, we infer that 
\begin{equation}|\Gamma(it)|^2 = \Gamma(it) \Gamma(-it) = \frac{\pi}{t \sinh(\pi t)}.\label{Gammait}\end{equation}
Thus, for $t=\frac{2\pi}{-\ln p}$, this gives
\begin{align}\label{sinh}
\sup_{x\in\R^+} |Q(x)| & \leq 2 \sum_{k\geq 1} \sqrt \frac{ \pi}{ k t \sinh(\pi k t)} = 
 \sqrt{\frac{\ln(1/p)}{2}} \sum_{k\geq 1} \sqrt \frac{1}{ k \sinh(\pi k t)}.
\end{align}
As, for $x\geq 0$, we have $\sinh(x) \geq (1/4) x \exp(4x/5)$, we get
\begin{align}
\sup_{x\in\R^+} |Q(x)| & \leq 
 \sqrt{\frac{\ln(1/p)}{2}} \sum_{k\geq 1} \sqrt \frac{1}{ (1/4) \pi k^2 t \exp(4\pi k t/5) } \notag
\\ 
&= {\ln(1/p)} \sum_{k\geq 1} \frac{1}{\pi k\exp\left(\frac{2}{5}\pi k t\right)} \label{sinh11}
\\
&= \frac{\ln(p)}{\pi}\ln\left(1-\exp\left(\frac{4\pi^2}{5\ln(p)}\right)\right). \label{sinh2}
\end{align}
Note that the more relaxed bound~\eqref{sinh2} is quite close to the stricter bound~\eqref{sinh}:
e.g.~for $p=1/2$ the bound~\eqref{sinh} gives the upper bound $1.090430\dots\times 10^{-6}$ (and one can numerically check that these first digits also constitute a \textit{lower} bound),
while the bound~\eqref{sinh2} gives the upper bound $2.49 \times 10^{-6}$.

Now, for bounding $R(x)$, we use the identity $\Gamma'(z)=\psi(z)\Gamma(z)$, with the bound~\eqref{bound1} from Lemma~\ref{LemmaPsi} for $|\psi(it)|$, %(with $t=-2k\pi/\ln(p)$),
and the bound~\eqref{sinh11} for $|\Gamma(it)|$:
\begin{align}
|R(x)|&\leq 2 \sum_{k\geq 1} \left|\Gamma'(s_k)\right| =2 \sum_{k\geq 1}\left|\psi(s_k)\right| \left|\Gamma(s_k)\right| \notag \\
&\leq \sum_{k\geq 1} \left( \frac{1}{2} \ln\left(1\!+\!\left(\frac{2\pi k}{\ln(p)}\right)^2\right) \! +\! \frac{\pi}{2}\!+\!1\!-\!\gamma\!-\!\frac{\ln p}{2\pi k} \right) 
\frac{\ln(1/p)} {\pi k\exp\left(-\frac{4}{5}\pi^2 k/\ln(p) \right)}.\qquad \label{summands}
\end{align}

Now, it is easy to check that we have $\frac{1}{2}\ln\left(1+x^2\right)\leq \exp\left(\frac{1}{31}\pi x\right)$ for all $x>0$. Then, noting $t=-2\pi/\ln(p)$, we get
\begin{align*}
 \sum_{k\geq 1} \frac{1}{2} \ln\left(1+(kt)^2\right) \frac{\ln(1/p)} {\pi k\exp\left(\frac{2}{5}\pi k t\right)} & \leq \sum_{k\geq 1} \frac{\ln(1/p)} {\pi k\exp\left(\frac{57}{155}\pi k t\right)} \\
&= \frac{\ln(p)}{\pi}\ln\left(1-\exp\left(\frac{114\pi^2}{155\ln(p)}\right)\right).
\end{align*}
Together with the contribution of the remaining summands in~\eqref{summands}, this gives the bound~\eqref{boundR} for $|R(x)|$.
\end{proof}

From this, we can establish the infinite differentiability of our fluctuations.
\begin{theorem}[Fourier series infinite differentiability]\label{Cinfty}
The Fourier series
\begin{equation*}
Q(x)=\sum_{k \in \Z\setminus \{0\}} \Gamma(s_k) \exp(-s_k x) 
\text{\qquad and \qquad }
R(x)=\sum_{k \in \Z\setminus \{0\}} \Gamma'(s_k) \exp(-s_k x)\end{equation*}
 (where $s_k=\frac{2ik\pi} {\ln p}$) are infinitely differentiable on ${\mathbb R}$.
\end{theorem}
\begin{proof}
A Fourier series $f(x)=\sum_{k \in {\mathbb Z}} c_k \exp(-i kx)$ satisfies the Weierstrass $M$-test if 
there exists a sequence $M_n$ such that $|c_k \exp(- i kx)|+|c_{-k} \exp(ikx)|<M_k$ 
(for all $x \in \mathbb R$) and $\sum_{k\geq 0} M_k$ converges.
If $f(x)$ and $g(x):= -i\sum_{k \in {\mathbb Z}} k c_k \exp(- ikx)$ 
both satisfy the Weierstrass $M$-test,
then they converge absolutely and uniformly in $\mathbb R$, and %we have 
$f'=g$. 

Thus, by successive application of this $M$-test, 
if the coefficients decay polynomially, i.e.,~we have $|c_{-k}|+|c_k|=O(|k|^{-d-1})$,
then $f(x)$ is in ${\mathcal C}^d$ (that is, $d$ times differentiable) and $f(x)$ is in ${\mathcal C}^\infty$
(that is, infinitely differentiable) if its coefficients decay faster than any polynomial rate.
By Equation~\eqref{Gammait}, the coefficients $\Gamma(s_k)$ decay like $\approx \exp(-k\pi/\ln(p))$,
so $Q(x)$ is in ${\mathcal C}^\infty$.
By Equation~\eqref{summands}, the coefficients $\Gamma'(s_k)$ also decay like 
an exponential, so $R(x)$ is in ${\mathcal C}^\infty$.
\end{proof}

It is interesting to compare this smoothness result 
with the situation observed by Delange~\cite{Delange1975} in his seminal work on the sum of digits of $n$ in base $1/p$
(when $1/p$ is an integer). Therein, he proved an asymptotic behavior involving fluctuations dictated by a Fourier series, 
which can also be obtained by a Mellin transform approach, quite similarly to the road followed in our article.
It appears that his Fourier series (already mentioned in Remark~\ref{rem:Fourier})
has coefficients $\zeta(s_k)/((1+s_k)s_k)\approx k^{-1.5}$;
it is thus not surprising that the Delange series is nowhere differentiable,
in sharp contrast with the smoothness of our Fourier series (see Figure~\ref{fig:Fourier}).
\smallskip

This concludes our analysis of the height and the corresponding fluctuations.

\begin{figure}[b]
\begin{tabular}{ccc}
\setlength{\tabcolsep}{0pt}
\hspace{-8.1pt}
\begin{tabular}{c}
\includegraphics[width=.301\textwidth]{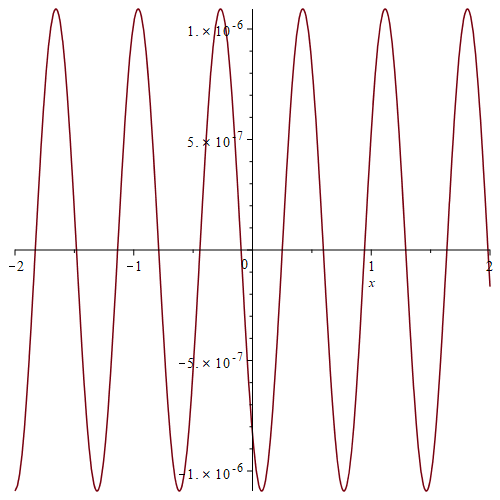}\\
$Q(x)$ (for $p=1/2$)
\end{tabular}
&\begin{tabular}{c}
\includegraphics[width=.301\textwidth]{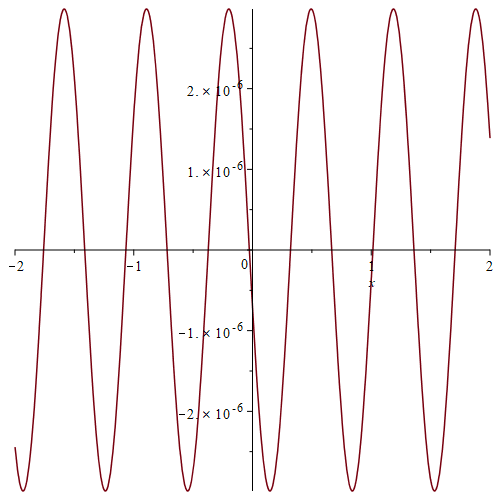} \\
$R(x)$ (for $p=1/2$)
\end{tabular}
&
\begin{tabular}{c}
\includegraphics[width=.301\textwidth]{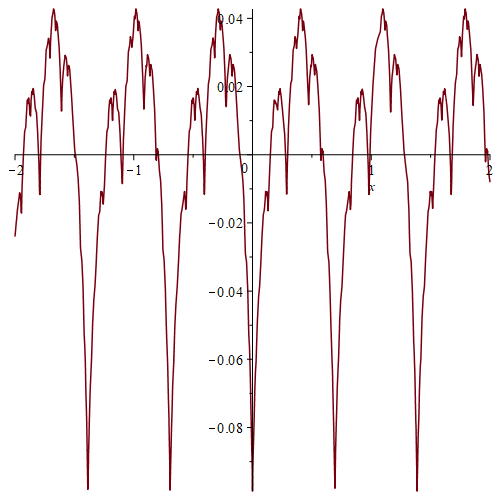}\hspace{-5pt}\\
$Delange(x)$ (for $p=1/2$)
\end{tabular}\hspace{-5pt}
\end{tabular}
\setlength\abovecaptionskip{2mm}
\caption{Our Fourier series $Q$ and $R$ are infinitely differentiable, 
while the Fourier series obtained by Delange is nowhere differentiable.
This follows from the asymptotics of their coefficients, as explained in the proof of Theorem~\ref{Cinfty}. %Delange is not C1,and  the ``nowhere'' claim uses one more argument.
}\label{fig:Fourier}
\end{figure}

\section{Some results for the Moran model in dimension \texorpdfstring{$m>1$}{m>1}}\label{Sec5}
\subsection{Joint distribution of ages for the Moran model with \texorpdfstring{$m>1$}{m=1}}
Moran processes are models of population evolution (or mutation transmission)
where the population is of constant size 
(some individuals could die but are then immediately replaced by a new individual). 
Depending on the applications, several variants were considered 
in the literature starting with the seminal work of Moran himself~\cite{Moran1958,Moran1962}, up to more 
recent extensions (for example to spatially structured population~\cite{LiebermanHauertNowak2005}.

Motivated by the model with resets of Itoh, Mahmoud, and Takahashi~\cite{ItohMahmoudTakahashi2004,ItohMahmoud2005},
we now define the \textit{Moran model with $m$ individuals}. 
It is a process parametrized by some probabilities $p$ and $p_i$'s such that $p+\sum_{i=0}^m p_i=1$,
and which starts at time 0 with $m$ individuals of age $0$.
Then, at each new unit of time,
\begin{itemize}
\item either, with probability $p$, all survive (their age increases by 1),
\item either, with probability~$p_i$ (for $1\leq i\leq m$), the $i$-th individual dies (it is then replaced by a new $i$-th individual of age 0), 
while the age of the $m-1$ surviving individuals increases by 1,
\item either, with probability $p_0$, all die and are replaced by $m$ new individuals of age~0.
\end{itemize}
Now, we define the sequence of multivariate polynomials $f_n(x_1,\dots,x_m)$
(for $n\in \N$)
by the fact that the coefficient of $x_1^{k_1} \cdots x_m^{k_m}$
in $f_n(x_1,\dots,x_m)$ is the probability that, at time~$n$, 
 the $i$-th individual has age $k_i$ (for $i=1,\dots, m$).
Accordingly, 
$F(t,x_1,\dots,x_m):= \sum_{n\geq 0} f_n(x_1,\dots,x_m) t^n$ is the 
probability generating function associated to the above Moran model,
 where the time is encoded by the exponent of $t$. 
\begin{theorem}\label{Th6}
The probability generating function of the Moran model is a rational function, and it admits the closed form\vspace{-.1mm}
\begin{equation}\label{Fmm}
F(t,x_1,\dots,x_m)=\frac{\sum_{k=0}^{2^m-1} (-1)^k P_k t^k}{\Delta},
\end{equation}
where the $P_k$'s are polynomials (given in the proof) in the $x_i$'s, $p$, $p_i$'s, 
and where $\Delta$ is the following polynomial of degree~$2^m$ in $t$:\vspace{-.1mm}
\begin{equation} \label{determinant} 
 \Delta= \prod_{I \subseteq \{1,\dots,m\}} \left( 1-t \left( p + p_0 [\![{I=\{1,\dots,m\}}]\!]+ \sum_{i\in I} p_i \right) \prod_{i\not \in I} x_i\right).\end{equation}
\end{theorem}
\begin{proof}
The Moran model evolution is encoded by the following functional equation for the probability generating function~$F$:\vspace{-.1mm}
\begin{align}\label{eqfunc}
F ( t, x_1, \dots, x_m ) = 1&+t p x_1 \cdots x_m F ( t, x_1, \dots, x_m )+ t p_0 F ( t,1,\dots,1 ) \nonumber \\
&+t \left( \sum_{i=1}^m p_i \frac{x_1 \cdots x_m}{x_i} F ( t, x_1, \dots, x_m )_{|x_i=1} \right), 
\end{align}
where $F_{|x_i=1}$ means $F$ evaluated at $x_i=1$.

To solve this single functional equation (which has $m+2$ unknowns\footnote{We temporarily count $F(t,1,\dots,1)$ as unknown, even if it is obviously equal to $1/(1-t)$, as $F$ is a probability generating function.}), 
the trick is to transform it into a linear system of equations with... $2^m$ unknowns!
Indeed, by substituting $x_i=1$ (in all the possible ways) in the functional equation~\eqref{eqfunc},
we get a system of $2^m$ equations. 
\pagebreak

Then, we encode this system by a matrix $M$, where we cleverly (sic!) choose 
the order in which unknowns are associated to the lines/columns of $M$.
Let us define this order; to this aim
consider the Cartesian product ${\mathcal X}:= \{1,x_1\} \times \cdots \times \{1,x_m\}$.
For any pair of $m$-tuples $\X$ and $\Y$ from ${\mathcal X}$, one writes $\X\prec \Y$
if the number of 1's in $\X$ is less than the number of 1's in~$\Y$,
or, when they have the same number of 1's, if $\X$ is smaller than~$\Y$ in the lexicographical order induced by $x_1 \prec \dots \prec x_m \prec 1$.
For example, we have $(x_1, x_2) \prec (x_1, 1) \prec (1,x_2) \prec (1,1)$.
Listing all the elements of $\mathcal X$ in increasing order, we get a list of $2^m$ tuples $X_1, \dots, X_{2^m}$.
The matrix $M$ encoding the aforementioned system of equations is constructed such that 
the $i$-th line of the matrix $M$ corresponds to the unknown $F(t,X_i)$ and the $j$-th column corresponds to the unknown $F(t,X_j)$. 

With this order, the matrix $M$ is an upper triangular matrix (as each of the substitution of some $x_i$'s by some 1's in Equation~\eqref{eqfunc} leads from some tuple $\X$ to $m+2$ larger tuples~$\Y$), and thus the determinant of $M$ is the product of its diagonal terms:
\begin{equation}
\det M = \ \prod_{I \subseteq \{1,\dots,m\}} \left( 1-t \left( p + p_0 [\![{I=\{1,\dots,m\}}]\!]+ \sum_{i\in I} p_i \right) \prod_{i\not \in I} x_i\right),\end{equation}
where we use Iverson's bracket notation\footnote{This notation, $[\![ \text{assertion} ]\!]$, is 1 if the assertion is true, and 0 if not. It was introduced in the semantics of the language 
APL by its founder, Kenneth Iverson. It was later popularized in mathematics by Graham, Knuth, and Patashnik~\cite{GrahamKnuthPatashnik1994}.}. 

As this determinant $\Delta:=\det M$ is not zero, this entails by Cramer's rule that $F ( t, x_1, \dots, x_m )$ can be written as a rational function with denominator $\Delta$ 
(note that, for some specific real values of $p$ and the $p_i$'s, it is not excluded that the numerator could have a shared factor with $\Delta$).
Of course, computing the determinant of each comatrix, and using the relation $p_0=1-(p+p_1+\dots+p_m)$, we get
symmetric polynomial expressions for the $P_k$'s occurring in~\eqref{Fmm}, e.g.:
\begin{align*}
P_0&=1, \nonumber \\
P_1&= p \left( \prod_{i=1}^m (1+x_i) - \prod_{i=1}^m x_i \right) + \sum_{i=1}^m x_i \sum_{\substack{j=1,\dots,m\\j\neq i}} p_j,\\
%P_1&=\sum_{I \subset \{1,\dots,m\}} \left( p + \sum_{i\in I} p_i \right) \prod_{i\not \in I} x_i,\\
\vdots \nonumber \\
P_{2^m-1}&=\left( \prod_{i=1}^m x_i^m \right) \prod_{I\subsetneq \{1,\dots,m\}} \left( p+ \sum_{i\in I} p_i \right). \qedhere
\end{align*}\end{proof}

Note that the case $p_0=0$, $p_i=1/m$ for $i=1,\dots,m$ (with $m\geq 2$) was analyzed by Itoh and Mahmoud~\cite{ItohMahmoud2005}: 
they proved that the age of each individual converges to
a shifted geometric distribution, namely $\operatorname{Geom}(1/m)-1$. They also show that the number 
of individuals of age~$k$ at time $n$ converges to a Bernoulli distribution,
namely $\operatorname{Ber}((m/(m-1))^k)$.\linebreak
Our Theorem~\ref{Th6} constitutes a joint law version of these results, at discrete times, for generic $p_i$'s.
For example, introducing $G(t,v):= \sum_{j=1}^m \binom{m}{j} v^j [x_1^k \dots x_j^k ] F(t,x_1,\dots,x_m)$,
the coefficient $[t^n] \partial_v G(t,1)$ gives the average number of individuals of age $k$ at time $n$.
(Note that the sum with the binomial coefficients $\binom{m}{j}$ has to be replaced by a sum over 
the subsets of $\{1,\dots,m\}$ if the $p_i$'s and the initial conditions for the $x_i's$ are not symmetric.)

\subsection{A multidimensional generalization of the Moran model}
Interestingly, the same strategy of proof allows us to solve a wide generalization of the Moran model, where 
\begin{itemize}
\item with probability~$p_I$, all the individuals from the subset $I$ of $\{1,\dots,m\}$ 
die (they are then replaced by new individuals of age 0), 
while the age of each surviving individual increases by 1.
\item the process starts with $m$ individuals of any (possibly distinct) ages, encoded by a monomial $f_0(x_1,\dots,x_m)$.
\end{itemize}

This translates to the following single functional equation, involving $2^m$ unknowns:
\begin{equation}
F(t,x_1,\dots,x_m) =f_0(x_1,\dots,x_m) + t\sum_{I\in \{1,\dots,m\}} p_I F(t,{\mathbf X}_I) \ \prod_{i\not\in I} x_i,\label{eqFI}
\end{equation}
\noindent where ${\mathbf X}_I = (x_1,\dots,x_m)_{|\text{$ x_i=1$ for all\ $i \in I$}}$.

Obviously, by taking $f_0=1$, $p_\varnothing=p$, $p_{\{1,\dots,m\}}=p_0$,
$p_{\{i\}}=p_i$, and all other $p_I=0$, the generalized model simplifies to the classical Moran model of Theorem~\ref{Th6}.
Another natural set of probabilities is $p_I= q^k (1-q)^{m-k}$, where $k$ is the number of elements in $I$.
It encodes the model where, at each unit of time, each individual dies with probability $q$.

More generally, for any set of $p_I$'s, one gets the following result. 
\begin{theorem}\label{Th7}
The probability generating function of the generalized Moran model is a rational function:
\begin{equation}\label{Fm}
F(t,x_1,\dots,x_m)=\frac{\sum_{k=0}^{2^m-1} (-1)^k Q_k t^k}{\Delta},
\end{equation}
where the $Q_k$'s are polynomials in the $x_i$'s and $p_I$'s for $I\subset\{1,\dots,m\}$,
and where $\Delta$ is the following polynomial of degree~$2^m$ in $t$:
\begin{equation}
 \Delta= \prod_{I \subseteq \{1,\dots,m\}} \left( 1-t \left( p_\varnothing + p_{\{1,\dots,m\}} [\![{I=\{1,\dots,m\}}]\!]+ \sum_{i\in I} p_{\{i\}} \right) \prod_{i\not \in I} x_i\right).\end{equation}
\end{theorem}
Note that, for this generalized model, the denominator $\Delta$ is the same as in Theorem~\ref{Th6},
and the $Q_k$'s are a lifting of the $P_k$'s from Theorem~\ref{Th6}, involving more terms and variables (namely, all the $p_I$'s).
For these two models, these polynomials $P_k$ and $Q_k$ are variants of symmetric functions. We comment more on this fact now.

\smallskip

\begin{remark}[Links with bi-indexed families of symmetric functions]\label{bi}
Many problems related to lattice paths lead to generating functions expressible in terms of symmetric functions;
this results from the kernel method, which involves a Vandermonde-like determinant, 
and thus leads to variants of Schur functions~\cite{Banderier2001,Bousquet-Melou08,BanderierLacknerWallner2020}. 
For the generalized Moran model we also get symmetric expressions, 
as the problem is by design symmetric, but in a more subtle way:
one does not get formulas  nicely expressible in terms of classical symmetric functions. 
This is due to the fact that we have to play with two distinct sets of variables (the $p_i$'s and the~$x_i$'s), 
the occurrences of which are not fully independent.
It appears that these\linebreak  subtle dependencies are well encoded 
by the \textit{MacMahon elementary symmetric functions}, defined by $e_{j,k} := [t_x^j t_p^k] \prod_{i=1}^m (1+t_x x_i+t_p p_i)$.
For example, we have $e_{2,1}= x_1x_2p_3 + x_2x_3p_1 + x_3x_1p_2$.
They allow us to provide more compact formulas for our generating functions, 
like $P_1= e_{1,1} +p \sum_{j=1}^m e_{j,0}$. We plan to study these aspects in a forthcoming work.
%Note that the group dictating their inner symmetries is not $S_{2m}$, but a diagonal action of the symmetric group.  
Note that these MacMahon symmetric functions also appear in problems a priori unrelated to our multidimensional Moran walks, 
see e.g.~the articles of Gessel~\cite{Gessel1988} and Rosas~\cite{Rosas2001}. 
\end{remark}

\subsection{Application to the soliton wave model}
The soliton wave model (as considered by Itoh, Mahmoud, and Takahashi~\cite{ItohMahmoudTakahashi2004})
is a stochastic system of particles encoding a unidirectional wave.
The number of particles is constant during the full process:
we have $m$ particles on $\Z$ which can only moves to the left as follows.
 At time $n=0$, the initial configuration consists of $m$ particles, at $x$-coordinates $1,\dots,m$. 
Then, at each unit of time $n = 1,\,2,\dots$, uniformly at random, one of the $m$ particles jumps just to the left of the first particle (the wave front),
thus leaving an empty space at its starting position:
\begin{center}
\includegraphics[width=4cm]{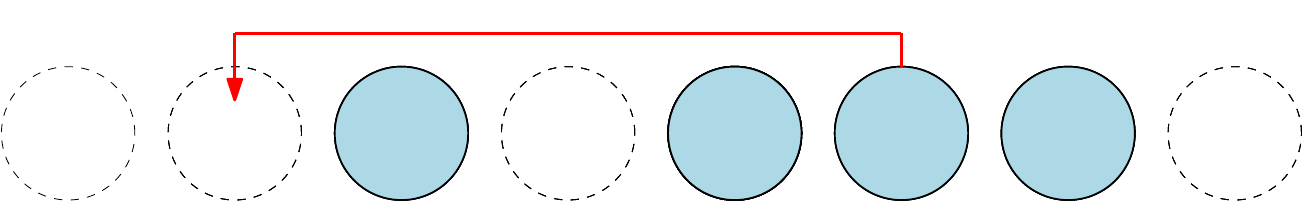} $ \ \longrightarrow \ $ \includegraphics[width=4cm]{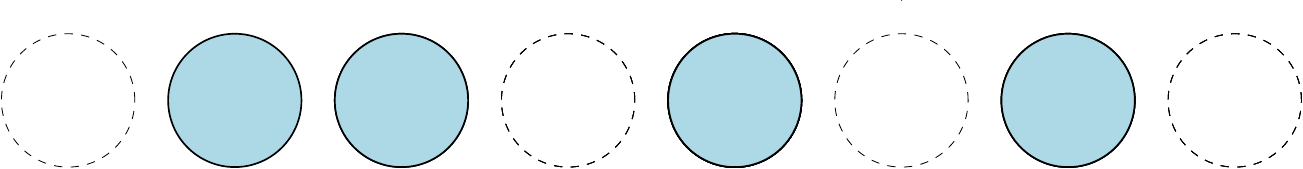} 
\end{center}

Note that at time $n$ the location of the leftmost particle has thus $x$-coordinate $1-n$. 
See Figure~\ref{soliton} for an illustration of 6 iterations of this process, where, for drawing convenience, 
we shift the origin of the $x$-axis after each step, so that the first particle is always at $x$-coordinate 1.

Then, applying Theorem~\ref{Th7} to this model, we get the following proposition.
\begin{proposition}
The joint distribution $F(t,x_1,\dots,x_m)$ of the time/positions of the particles in the soliton wave model is given by Formula~\eqref{Fm},
by taking as initial condition $f_0=x_1^1 x_2^2 \dots x_m^m$, 
and, as probabilities of transition, $p_{\{i\}}=1/m$ and all other $p_I=0$;
what is more, the denominator of $F(t,x_1,\dots,x_m)$ thus simplifies to 
\begin{equation*}
 \Delta= \prod_{I \subseteq \{1,\dots,m\}} \left( 1-t\frac{|I|}{m} \right) \prod_{i\not \in I} x_i,
\end{equation*}
where $|I|$ stands for the number of elements of the set $I$.
\end{proposition}
\vspace{-3mm}

\newcommand{\myW}{5.3cm}
\newcommand{\myH}{5.9mm}
\begin{figure}[b]
\setlength\belowcaptionskip{0mm}
\setlength\abovecaptionskip{2mm}
\renewcommand{\arraystretch}{1}
\begin{tabular}{|c|c|c|}
\hline 
 Wave& Time $n$ & Length $L_n$ \\
 \hline\setlength\extrarowheight{\myH}
\begin{tabular}{c}\includegraphics[width=\myW]{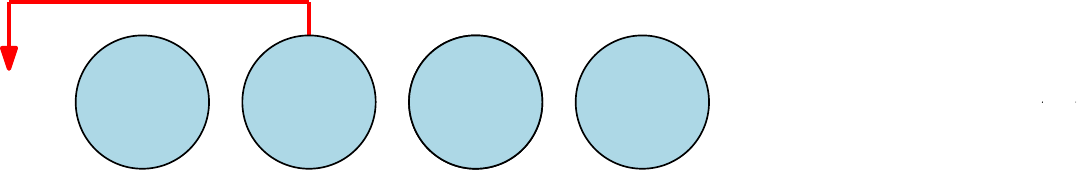}\end{tabular} & 0 &4\\
 \hline\setlength\extrarowheight{\myH}
\begin{tabular}{c}\includegraphics[width=\myW]{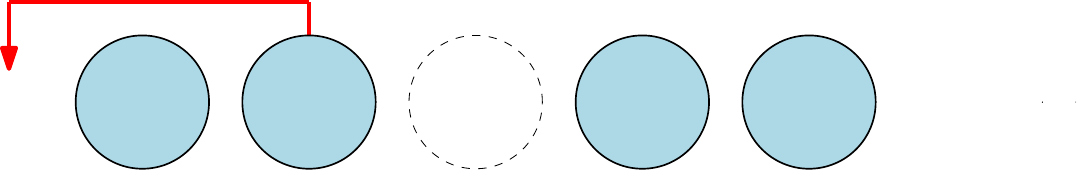} \end{tabular}&1 & 5 \\
 \hline\setlength\extrarowheight{\myH}
\begin{tabular}{c} \includegraphics[width=\myW]{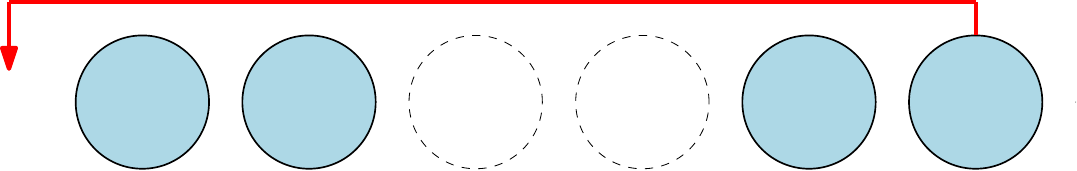} \end{tabular} &2 & 6 \\
 \hline\setlength\extrarowheight{\myH}
\begin{tabular}{c} \includegraphics[width=\myW]{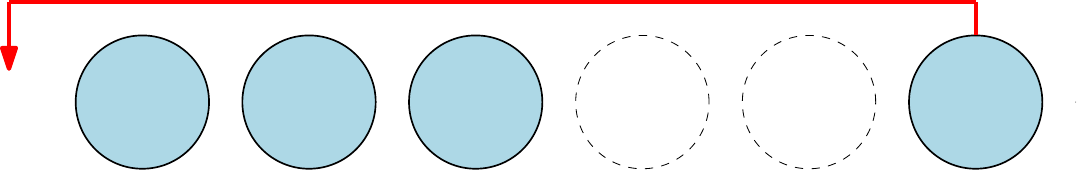} \end{tabular}&3 & 6 \\
 \hline\setlength\extrarowheight{\myH}
\begin{tabular}{c}\includegraphics[width=\myW]{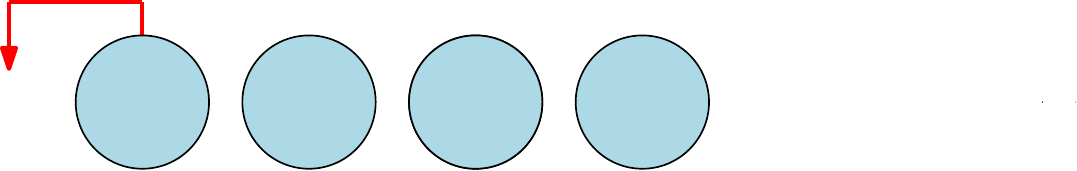} \end{tabular} &4 & 4 \\
 \hline\setlength\extrarowheight{\myH}
\begin{tabular}{c}\includegraphics[width=\myW]{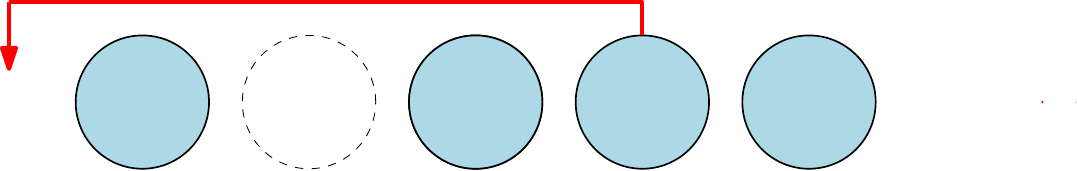} \end{tabular} &5 & 5 \\
 \hline\setlength\extrarowheight{1mm}
\begin{tabular}{c}\includegraphics[width=\myW]{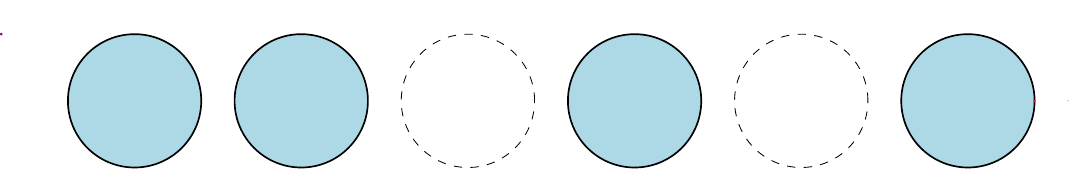} \end{tabular} &6 &6 \\
 \hline
\end{tabular}
\caption{The soliton wave model: 
a wave is a sequence of particles (the sequence may have some inner holes), 
and at each unit of time, 
one particle is selected and jumps at the very start of the wave (and thus leaves an empty slot where it was). 
Trailing empty slots are ignored (this occurs when the last particle is selected, e.g.~from step 3 to 4 above).
}\label{soliton}
\end{figure}

Figure~\ref{soliton} also shows that this model has one degree of freedom, that is,
the soliton wave model with $m$ particles can be modeled as $m-1$ interactive urns $U_1, \dots, U_{m-1}$:
the urn $U_k$ contains the number of white cells between the $k$-th and $(k+1)$-th blue particle. 
Accordingly, this interactive urn process starts with $U_k(0)=0$ for all $k$,
and then, at each unit of time, we have one of the following $m$ events (with probability $1/m$):\vspace{-1.3mm}

\setitemize{labelindent=8mm,labelsep=2mm,leftmargin=*}
\begin{itemize}
\item $U_1(n+1)=U_1(n)+1$ and other urns are unchanged. %first particle chosen
\item for $k=2,\dots,m-1$: $U_1(n+1)=0$, $U_j(n+1):=U_{j-1}(n)$ (for $j=2,\dots,k-1$), 
 $U_k(n+1):=U_{k-1}(n)+U_{k}(n)+1$, and remaining urns are unchanged. % k-th particle chosen k =2...m-1 
\item $U_1(n+1)=0$ and, for $k\geq 2$, $U_k(n+1):=U_{k-1}(n)$. %m-th particle chosen
\end{itemize}   %N.B.: this is not a P\'olya urn process.
The length of the soliton is then given by
$L_n=m+U_1(n)+\dots+U_{m-1}(n)$;
it can equivalently be viewed as the maximum  of the $x$-coordinates (at time $n$) of each particle.\vspace{-1mm}

\section{Conclusion and future works}\label{Sec6}

In this article, we considered several statistics (final altitude, waiting time, height) associated to walks with resets, for any given finite step set.
For the case of the simplest non-trivial model (namely, for Moran walks),
we prove that the asymptotic height exhibits some subtle behavior related to the discrete Gumbel distribution.
In a forthcoming article, we plan to consider the asymptotic analysis of the height for more general walks. 

In our formulas for walks of length $n$, taking $q':=q/n$ (and more generally $q'=q(n)$) as the probability of reset
leads to models which can counterbalance the infinite negative drift of the initial model,
and thus present a different type of asymptotic behavior. 
Studying these models and their phase transitions in more detail would be interesting.

In Section~\ref{Sec5}, we considered several multidimensional extensions of such walks,
with applications to the soliton wave model, or to models in genetics.
More multidimensional variants of Moran models allowing both positive and negative jumps
(and with or without resets) can be handled using the approach presented in this article (see~\cite{Althagafi2023}).
One interesting example is the one where each dimension evolves like a Motzkin path, 
this model was e.g.~considered in the haploid Moran model~\cite{HuilletMoehle2009},
where the authors use a Markov chain approach, using duality/reversibility 
to establish links with Ornstein--Uhlenbeck processes.
Note that even if one adds resets to such Motzkin-like models, one keeps nice links with continuous fractions associated to birth and death processes; see~\cite{FlajoletGuillemin2000}.
The analysis becomes much more complicated as soon as jumps of amplitude $\geq 2$ are allowed;
in such cases, our approach based on the kernel method strikes again.

Another natural extension is to consider walks in the quarter plane with resets (a natural model of two queues evolving in parallel);
even for walks with jumps of amplitude~$1$, the exact enumeration and the asymptotic behavior of the (maximal) height remain open.
Other more ad hoc extensions consider some age-dependent probabilities $p_i$'s, then leading 
to partial differential equations for the corresponding generating functions. Some specific cases lead to closed-form solutions. 

All these variants of Moran models are parametrized by the $p_i$'s. 
One can then turn to the tuning of several statistical tests: having some experimental data, 
it is natural to look for maximum likelihood estimators of the $p_i$'s, 
and to study if they are unbiased, sufficient, and consistent (for more on these notions, see e.g.~\cite{Spanos2019}). 
In conclusion, the Moran model offers a large variety of  interesting models, with many aspects to explore!

\noindent{\bf Acknowledgement:} The work of the first author is supported by the Researchers Supporting Project RSPD2023R987 
of King Saud University.
We thank Hosam Mahmoud for introducing us to the Moran walks, 
and asking us about the distribution of their height.
We are indebted to Rosa Orellana and Mike Zabrowski for pinpointing us that the bi-indexed symmetric functions which we introduced
in Remark~\ref{bi}, and which already occurred in the literature under the name \textit{MacMahon symmetric functions}.
Last but not least, we also deeply thank the two referees for their kind detailed reports, which helped to improve several parts of this article.

\bibliographystyle{SLC}
\bibliography{AguechAlthagafiBanderier}
\end{document}